\documentclass[11pt]{article}
\usepackage[margin=1in]{geometry}

\usepackage{mdframed}
\usepackage{subfigure}
\usepackage{cite}
\usepackage{stfloats}
\usepackage{relsize}

\usepackage{setspace}
\usepackage{amsopn}
\usepackage{float}
\restylefloat{table}
\usepackage[font={small,it}]{caption}
\usepackage{algpseudocode,algorithm,algorithmicx}
\usepackage{amsxtra}
\usepackage{amsmath}
\usepackage{graphicx}
\usepackage{latexsym}
\usepackage{amsfonts}
\usepackage{amssymb}
\usepackage{bm}
\usepackage{yhmath}
\usepackage{verbatim}
\usepackage{float}
\usepackage{xr-hyper}
\usepackage{multirow}
\usepackage{subfloat}

\usepackage{accents}

\usepackage{epstopdf}

\usepackage{mdframed}
\usepackage{subfigure}
\usepackage{cite}
\usepackage{stfloats}
\usepackage{relsize}

\usepackage{setspace}
\usepackage{amsopn}
\usepackage{float}
\restylefloat{table}
\usepackage{amsthm}
\usepackage[font={small,it}]{caption}
\usepackage{algpseudocode,algorithm,algorithmicx}
\usepackage{amsxtra}
\usepackage{amsmath}
\usepackage{graphicx}
\usepackage{latexsym}
\usepackage{amsfonts}
\usepackage{amssymb}
\usepackage{mathrsfs}
\usepackage{bm}
\usepackage{yhmath}
\usepackage{verbatim}
\usepackage{float}
\usepackage{xr-hyper}
\usepackage{multirow}
\usepackage{subfloat}

\usepackage{epstopdf}

\usepackage{authblk}
\newtheorem{proposition}{Proposition}[section]

\theoremstyle{definition}
\newtheorem{definition}{Definition}[section]

\theoremstyle{condition}

\DeclareMathAlphabet{\bi}{OML}{cmm}{b}{it}
\DeclareMathAlphabet{\bcal}{OMS}{cmsy}{b}{n}
\DeclareMathAlphabet{\brmn}{OT1}{cmr}{bx}{n}

\DeclareMathSymbol{\R}{\mathalpha}{AMSb}{"52}

\def\x{\mathbf{x}}
\def\X{\mathbf{X}}

\def\Y{\mathbf{Y}}

\def\A{{\cal A}}

\def \a{\mathbf{a}}
\def \v{\mathbf{v}}
\def \y{\mathbf{y}}
\def \z{\mathbf{z}}

\def \p{\mathbf{p}}

\def \d{\mathbf{d}}

\def \q{\mathbf{q}}
\def \s{\mathbf{s}}

\title{A Spectral Estimation Framework for Phase Retrieval via Bregman Divergence Minimization}

\author{Bariscan Yonel, Birsen Yazici}
\affil{Electrical \& Computer Systems Engineering, Rensselaer Polytechnic Institute, Troy, NY}
\affil[]{\{yonelb@rpi.edu, yazici@ecse.rpi.edu \}}

\begin{document}
%\ninept
%
\maketitle
\begin{abstract}
In this paper, we develop a novel framework to optimally design spectral estimators for phase retrieval given measurements realized from an arbitrary model. 
We begin by deconstructing spectral methods, and identify the fundamental mechanisms that inherently promote the accuracy of estimates. 
We then propose a general formalism for spectral estimation as approximate Bregman loss minimization in the range of the lifted forward model that is tractable by a search over rank-1, PSD matrices. % such that a Bregman loss is approximately minimized 
Essentially, by the Bregman loss approach we transcend the Euclidean sense alignment based similarity measure between phaseless measurements in favor of appropriate divergence metrics over $\mathbb{R}^M_+$. 
%In essence, these methods search for an estimate via a similarity measure over phaseless measurements which is conventionally the Euclidean sense alignment over $\mathbb{C}^M$. 
%By the Bregman loss approach, the similarity measure is defined for phaseless measurement over more appropriate domains in $\mathbb{R}^M$.
To this end, we derive spectral methods that perform approximate minimization of KL-divergence, and the Itakura-Saito distance over
phaseless measurements by using element-wise sample processing functions.
%We further relate our framework to the optimal sample processing functions derived in the literature under the specificity of the Gaussian sampling model.  
As a result, our formulation relates and extends existing results on model dependent design of optimal sample processing functions in the literature to a model independent sense of optimality.
Numerical simulations confirm the effectiveness of our approach in problem settings under synthetic and real data sets. %, including the public optical imaging data set provided by Metzler et.al.
%The proposed framework consistently outperforms the state-of-the art initialization methods in empirical performance without increasing computational complexity. 

\end{abstract}

\section{Introduction}

Phase retrieval literature has witnessed a major resurgence in the last decade by an influx of methods equipped with theoretical performance guarantees %for solving quadratic systems of equations, 
despite the ill-posed nature of the problem. 
Most of the compelling theoretical developments in the field rely on certain properties 
%some of which were studied in this thesis, that 
that are satisfied by statistical models such as Gaussian sampling, and coded diffraction patterns with high probability and are relevant to problems in wide range of applications such as optical imaging and astronomy, blind channel estimation and equalization, X-Ray crystallography, quantum state tomography, and subspace tracking. 

In this paper, we assume the setting of a \emph{starved regime}, which requires successfully recovering an object of interest $\x \in \mathbb{C}^N$ given $M = \mathcal{O}(N) \ll \mathcal{O}(N^2)$ (realistically near the information theoretic limits of injectivity $M \geq 4N-4$ \cite{bandeira2014saving}) intensity only measurements of the form:
\begin{equation}\label{eq:phaless}
y_m = | \langle \a_m, \x \rangle |^2, \quad m = 1, 2, \cdots M, 
\end{equation} 
where $\a_m \in \mathbb{C}^N$ denotes the $m^{th}$ sampling or measurement vector.  

Until recently, the state-of-the-art approaches to exact phase retrieval were lifting-based, convex semi-definite programming methods which solved the inversion from \eqref{eq:phaless} over the positive semi-definite cone in $\mathbb{C}^{N \times N}$, while attaining low sample complexities by leveraging the rank-1 structure of the lifted unknown $\mathbf{X} = \mathbf{x} \mathbf{x}^H$. 
However, the practical drawbacks such as increased computational complexity and memory requirements imposed by the high dimensional search space of the problem lead to the study of algorithms that operate on the original signal domain.  %despite the ill-posed, non-convex nature of the problem. 

On the other hand, algorithms that operate on the signal domain employ a variety of techniques to deal with the non-convexity of the problem which arise from having quadratic equality constraints. 
These techniques are based on the first principle of fielding an accurate initial estimate such that the algorithm is guided to a neighborhood of the search space, i.e., a basin of attraction around a global solution. 
For the typically studied statistical measurement models, $\ell_2$-mismatch in the range of the quadratic measurement map in \eqref{eq:phaless} defined as:
\begin{equation}\label{eq:objfun}
f(\z) = \frac{1}{2M} \sum_{m = 1}^M  ( | \langle \a_m, \z \rangle |^2 - y_m )^2,
\end{equation}
is smooth and has sufficient curvature, or satisfies restricted strong convexity within the basin of attraction obtained through an accurate initial estimate. 
These ultimately facilitate the convergence of steepest descent-type iterations on the signal domain, and subsequently, the exact recovery guarantees. 
Therefore, initialization step plays a crucial role in the formalism as well as the performance guarantees of celebrated, seminal signal domain algorithms for phase retrieval, such as the non-convex framework of Wirtinger Flow (WF) \cite{candes2015phase} and its variants \cite{zhang2016reshaped, chen2017solving, chen2017, Zhang2017a, wang2018phase, wang2018solving}, or convex relaxation approach of PhaseMax (PM) \cite{goldstein2018phasemax} and its linear programming based counterparts \cite{bahmani2017flexible, hand2016elementary}. 

For the purpose of exactly solving the original non-convex optimization problem, the most prominent initialization schemes in the literature are \emph{spectral methods}.  
The spectral initialization approach was first proposed in  \cite{netrapalli2013phase} to develop exact recovery guarantees for an alternating minimization approach to phase retrieval by stochastically re-sampling subsets of phaseless measurements at the initialization and in the gradient steps. 
Spectral methods were later adopted as provably good initial estimators in the subsequent exact recovery frameworks of WF and PM which eliminated re-sampling in their guarantees, where the accuracy of spectral initialization played an integral role on their analysis over the original signal domain \cite{netrapalli2013phase, candes2015phase, goldstein2018phasemax}. 
These developments sparked further studies which evaluate and accordingly improve the performance of spectral estimates by the means of sample truncation \cite{chen2017solving, chen2017, Zhang2017a}, design of sample processing functions \cite{lu2017phase, mondelli2017fundamental, luo2019optimal, gao2017phaseless}, or constructing alternative matrices for the spectral estimation of the ground truth, such as with linear spectral estimators \cite{ghods2018linear}, maximal correlation \cite{wang2018phase} and orthogonality promoting methods \cite{wang2018solving, chen2018phase}. 

Truncating high magnitude samples is originally motivated in \cite{chen2015solving, chen2017solving} for conditioning the spectral content of the weighted sample covariance matrix in order to achieve a $\mathcal{O}(N)$ complexity of measurements, rather than the $\mathcal{O}(N \log N)$ complexity of the Gaussian sampling model \cite{yonel2020deterministic}.  
More general sample processing functions were proposed instead of truncation, in which different optimality criteria are considered and achieved for the Gaussian sampling model to guarantee what is referred to as \emph{weak recovery} \cite{luo2019optimal, mondelli2017fundamental}, i.e., a high probability guarantee of a strictly positive correlation coefficient with the underlying ground truth at any sampling complexity above a critical threshold. 
A general class of linear spectral estimators were formulated in \cite{ghods2018linear}, in which the minimum mean squared error optimality criterion on the \emph{lifted} domain is considered which yields an analogue of filtering the phaseless measurements given the instance of a collection of measurement vectors prior to backprojection. %under the minimum mean squared error optimality criterion on the \emph{lifted} domain. 
Orthogonality promoting \cite{wang2017solving} or null-vector initialization methods \cite{chen2018phase, liu2019robust} rely on the near orthogonality properties of high dimensional i.i.d. Gaussian sampling vectors to a ground truth independently sampled from the unit sphere on $\mathbb{C}^N$ with high probability, and extract  
%To this end these methods propose extracting 
an estimate from the minimal eigenvector of the empirical covariance matrix of the sampling vectors.

In this paper, %we take an entirely novel approach to design of spectral methods, where we develop a phaseless estimation framework for minimizing general Bregman divergences
we take an entirely novel approach for the design of spectral methods.
Particularly, we interpret the initialization procedure under a phaseless estimation framework where the spectral method formulation is analogous to an exhaustive search over the non-convex set of rank-1, PSD matrices, which is conducted \emph{tractably} while approximating a loss minimization objective in the range of the lifted forward model of phaseless measurements. 
To this end, we first deconstruct the classical formulations in the literature to identify an inherent mechanism of $\ell_2$ loss minimization over phaseless measurements.
This minimization is approximately conducted by the spectral method under the restricted isometry-type property of \cite{yonel2020deterministic}, as a sufficient condition for exact recovery.
Equivalently, we interpret that the classical formulation is expected to promote an accurate estimate through an Euclidean sense alignment in the range of the lifted forward model in the absence of a proper normalization of phaseless measurements. 
Using this approach, we motivate the sample processing functions pursued in the literature in a unifying perspective of improving the similarity measure over the phaseless measurements.
Subsequently, we derive fundamental principles of phaseless estimation under universal arguments that pertain to generic measurement maps instead of model specific arguments predominantly encountered in the literature. 
We achieve such a formalism by establishing an abstract framework for initial estimation while utilizing \emph{Bregman divergences}, where the sample processing function has the role of \emph{warping} the underlying similarity measure of phaseless measurements to alternative forms. 
Through fundamental property of \emph{Bregman representation}, we establish that a sample processing function can be designed \emph{optimally} given a Bregman divergence of choice, such that the estimator minimizes a surrogate that is formulated under \emph{minimal distortion} from the original objective.  
%Thereby, the role of the sample processing function is characterized as \emph{warping} the similarity measure over phaseless measurements to particular Bregman divergences.  

%In essence, the developed framework is analogous to a variational formulation over a function space, such that for any Bregman divergence, the best approximation that admits a tractable spectral method is constructed, and accordingly used to form an initial estimate.  

%Bregman divergences are invaluable tools in science and engineering. 
%, and \emph{iii}) synthetic 2D-Fourier slide measurements emulating far-field electromagnetic scattering under the Born approximation in relevance to applications in wave-based imaging. 
%For our simulations we consider three settings, 
Our use of Bregman divergences are motivated through their universal properties that by definition generalize to any specification of a strictly convex, differentiable function. 
Indeed for this reason, Bregman's characterization in \cite{bregman1967relaxation} has found extensive use in fields of machine learning \cite{banerjee2005clustering, frigyik2008functional, reid2009surrogate, abernethy2012characterization, reid2011information}, information theory and statistical signal processing \cite{cover1999elements, gray1980distortion, gray1981rate, banerjee2005optimality}, and optimization \cite{censor1981iterative, bauschke2001joint}.
%In machine learning, Bregman divergences were studied for clustering \cite{banerjee2005clustering}, model boosting \cite{lafferty1999additive, kivinen1999boosting}, regret and risk analysis in statistical learning \cite{nock2008bregman, reid2009surrogate, reid2011information}. % decision theory \cite{} and Bayesian estimation \cite{}.
Notably, due to their generic structure, Bregman divergences were proved to be an invaluable tool for problems in signal processing such as compressed sensing and image reconstruction, as they were utilized in the implementation of iterative regularization with total variation and wavelet-based denoising algorithms \cite{yin2008bregman, cai2009linearized}. 
Specific instances of Bregman divergences have also been imperative in various applications in engineering, such as Kullback-Leibler (KL) divergence in communication systems \cite{kullback1997information}, and Itakura-Saito (IS) divergence in speech recognition systems \cite{itakura1975minimum}. 
%In particular, the KL divergence is equivalently termed \emph{discrimination information} and has an intuitive meaning in information theory and coding \cite{}. 
%On the other hand IS divergence is known as the optimal
In particular, KL divergence is widely known as the \emph{discrimination information} between two probability distributions, while motivated as their natural dissimilarity measure through information theory and coding \cite{cover1999elements}.
Similarly, IS divergence is studied for spectral density estimation using auto-regressive models, while known to be asymptotically equivalent to the KL divergence of the ensemble probability distributions when discriminating the spectral densities of two zero mean Gaussian processes \cite{gray1980distortion, gray1981rate}. 
Furthermore, very recently during the development of this paper \cite{vial2020phase} proposed minimizing KL and IS divergence metrics for retrieving audio signals from phaseless STFT measurements using ADMM.

In the context of phase retrieval our motivation to study optimal designs for spectral methods is based on the observations from \cite{yonel2020deterministic}, which identify the classical spectral initialization as a limiting factor in the universal performance guarantees and robustness of the non-convex approach. 
Having strictly non-negative measurements in the form of \eqref{eq:phaless}, we consider the abstract initial estimation as an analogous problem to spectral density or probability density estimation, and pursue information theoretic metrics for improved performance over the conventional quadratic loss-based formulation. 
%This approach consolidates the fundamental principles of state-of-the art initialization methods under a loss minimization based perspective, which we utilize for minimizing 
%%We thereby extend these ideas onto 
%Bregman distances on $\mathbb{R}_+^M$ that are more suitable to assess similarity of phaseless measurements than the Euclidean distance measure.  
Accordingly,
using our formalism we derive sample processing functions that are optimal for minimizing the KL divergence and the IS divergence in the range of the lifted forward model through the spectral estimation procedure. 
Notably, we introduce a flexible and inclusive sense of optimality in the design of spectral methods that is metric specific, as opposed to the model specific characterizations in the literature. 
Furthermore, we establish our framework as a generalization of the optimal sample processing functions derived for the Gaussian sampling model in \cite{luo2019optimal}, which appear as an instance of our Itakura-Saito optimal formulation. 
%Furthermore, we establish a connection between our framework and the optimal pre-processing function derived specifically for the Gaussian sampling model, and generalize the optimality result to Itakura-Saito optimality under arbitrary measurement models 
We support our theoretical foundations with numerical simulations that demonstrate the effectiveness of our approach via the particular sample processing functions we develop for phase retrieval. %which compare favorably against the spectral methods in the literature. 
For our simulations, we consider two problem settings, \emph{i}) synthetically generated measurements from the Gaussian sampling model, and \emph{ii}) a real optical imaging problem using the data set provided by \cite{metzler2017coherent} which is sensed by solving the double phase retrieval problem \cite{hwang2009fast} for 2D imagery generated by a spatial light modulator.

The rest of the paper is organized as follows: Section \ref{sec:4_sec1} reviews classical spectral estimation and its sample-processing based extensions, and discusses their inherent mechanisms and limitations. Section \ref{sec:4_sec3} presents our formulation of Bregman distance minimizing spectral methods, and the proposed framework for optimally deriving sample processing functions. Section \ref{sec:4_sec4} presents particular instances of optimal sample processing functions derived under our framework for KL, and IS divergences, and discusses the relation of the framework to the optimal sample processing functions in the literature. Section \ref{sec:4_sec5} provides our numerical simulations and \ref{sec:4_sec6} concludes the paper. % with discussion on future directions for our work.   

\section{Spectral Methods for Phase Retrieval}\label{sec:4_sec1}
\subsection{Classical Method for Algorithm Initialization}
\subsubsection{Statistical formulation} The original spectral method generates an estimate for the phase retrieval problem as the leading eigenvector of the following positive semi-definite data matrix:
\begin{equation}\label{eq:spectral}
\mathbf{Y} = \frac{1}{{M}}\sum_{m = 1}^{{M}} y_{m} \mathbf{a}_m \mathbf{a}_m^H.
\end{equation}
%The leading eigenvector is scaled by the square root of the corresponding largest eigenvalue $\lambda_0$ of $\mathbf{Y}$.
Motivations of the spectral method are classically rooted in a stochastic perspective.
By the strong law of large numbers, under the assumption that $\mathbf{a}_m \sim \mathcal{N}(0, \frac{1}{2} \mathbf{I}) + \mathrm{i} \mathcal{N}(0, \frac{1}{2} \mathbf{I} ), m = 1, \cdots, M$, the matrix $\mathbf{Y}$ tends to
\begin{equation}
\mathbb{E} \big[ \mathbf{Y}  \big] = \| \x \|^2 \mathbf{I} +  \x \x^H,
\end{equation}
as $M \rightarrow \infty$, which has the true solution $\x$ as its leading eigenvector.
Intuitively, tight concentration of the empirically formed data matrix $\mathbf{Y}$ around its expectation implies an accurate alignment of its leading eigenvector $\v_0$ to the underlying ground truth as:
\begin{equation}\label{eq:wkrec}
\big\| \mathbf{Y} - (\| \x \|^2 \mathbf{I} +  \x \x^H) \big\| \leq \delta \quad \rightarrow \quad \rho(\x, \alpha) := \frac{| \langle \v_0 , \x \rangle |}{\| \x \|} \geq \varepsilon,
\end{equation}
where $\delta$ is a small constant, $\alpha = M/N$ is the oversampling factor, and
\begin{equation}\label{eq:specmet}
\v_0 =  \underset{\|\v \| = 1}{\text{argmax}} \ \v^H \mathbf{Y} \ \v. 
\end{equation}
Ultimately, the goal for using a spectral method is to achieve a sufficiently large $\varepsilon > 0$ to land within a basin of attraction by the initial estimate, such that the algorithm converges to a solution $\x e^{j \phi}$. 
To this end, \eqref{eq:wkrec} is closely related to the tightness of the following restricted isometry-type property of \emph{the lifted forward model} over the set of rank-1, positive semi-definite matrices as \cite{yonel2020deterministic}:
\begin{equation}\label{eq:tightframe}
(2 - \delta) \| {\z \z^H} \|_F^2 \leq \frac{1}{M} \| \mathcal{A}(\z \z^H) \|_2^2 \leq (2 + \delta) \| {\z \z^H} \|_F^2,
\end{equation}
with $\mathcal{A}: \mathbb{C}^{N \times N} \rightarrow \mathbb{C}^M$, where \eqref{eq:phaless} is equivalent to $ \mathcal{A}(\x \x^H) = \y \in \mathbb{R}^M$.

\subsubsection{Approximate $\ell_2$-loss minimization} Using the definition of the lifted forward model, the $\ell_2$-mismatch loss in \eqref{eq:objfun} is transformed into the familiar least-squares formulation as:
\begin{equation}\label{eq:objfun2}
\mathcal{J}(\X) = \frac{1}{2M} \| \mathcal{A}(\X ) - \y \|_2^2.
\end{equation}
Clearly $\X = \x \x^H$ is an optimal solution of \eqref{eq:objfun2}, however in typical inference problems $\mathcal{A}$ has a non-trivial null space as it characterizes a severely underdetermined system of equations with $M \ll N^2$. 
Nonetheless, spectral methods produce accurate estimators of the lifted ground truth $\tilde{\X} = \x \x^H$ under the least squares criterion in \eqref{eq:objfun2} given that $\mathcal{A}$ satisfies certain conditions such as the restricted isometry property provided in \eqref{eq:tightframe}. 

Under \eqref{eq:tightframe}, the spectral method practically becomes a search over the set of rank-1, PSD matrices with such that the $\ell_2$ loss is approximately minimized, as 
\begin{align}\label{eq:spec2}
\mathcal{J}(\z \z^H) &= \frac{1}{2M}\| \mathcal{A} (\z \z^H) \|_2^2 + \frac{1}{2M} \| \y \|_2^2 - \frac{1}{M}  \langle \mathcal{A}(\z \z^H) , \y \rangle, \\
&\approx \| \z \z^H \|_F^2 + \frac{1}{2M} \| \y \|_2^2 - \frac{1}{M}  \langle \mathcal{A}(\z \z^H) , \y \rangle,
\end{align}
if $\delta$ in \eqref{eq:tightframe} is a sufficiently small constant, where the inner-products are real valued from the definition of $\mathcal{A}$ and $\y \in \mathbb{R}_{+}^{M}$.  
Equivalently, using an \emph{a priori} estimate for the energy of the underlying signal by $\lambda_0$, we can write minimization of \eqref{eq:spec2} as
\begin{equation}
\underset{\| \z \| = \sqrt{\lambda_0}}{\text{minimize}} \ \mathcal{J}(\z \z^H) \approx \underset{\| \z \| = \sqrt{\lambda_0}}{\text{maximize}} \ \langle \z \z^H , \frac{1}{M} \mathcal{A}^H (\y )\rangle_F,
\end{equation}
using the definition of the adjoint operator
and the fact that $\langle \z \z^H , \frac{1}{M} \mathcal{A}^H (\y )\rangle_F$ term in \eqref{eq:spec2} is a positive real by definition. 
%Notably, such an \emph{a priori} estimate is available through the property in \eqref{eq:tightframe} if $\delta$ is sufficiently tight, or alternatively from an $\ell_1$ restricted-isometry of $\mathcal{A}$ where
%\begin{equation}\label{eq:l1_isom}
%(1 - \delta_0) \| \z \|_2^2 \leq \frac{1}{M} \| \mathcal{A}(\z \z^H) \|_1 \leq (1 + \delta_0) \| \z \|_2^2
%\end{equation}
%which is relevant to the statistical formulation under the Gaussian sampling model, as by definition
%\begin{equation}
%
%\end{equation} 

Having fixed the energy of the search variable the maximizer in \eqref{eq:spec2} is equivalent to the leading eigenvector of $\Y$ where
%From the definition of the Frobenius inner-product, the spectral method in \eqref{eq:specmet} corresponds to maximizing the final $\langle \z \z^H , \frac{1}{M} \mathcal{A}^H (\y )\rangle_F$ term in \eqref{eq:spec2} which is guaranteed to be positive. Fixing the first term as $\| \z \z^H \|_F^2 = \lambda_0$ using the estimator for the energy of the underlying signal and having
\begin{equation}
\langle \z \z^H , \frac{1}{M} \mathcal{A}^H (\y )\rangle_F = \z^H \left[\frac{1}{M} \mathcal{A}^H (\y )\right]\z,
\end{equation}
from the definition of Frobenius inner-product. 
The adjoint of $\mathcal{A}$ then equivalently defines \emph{the spectral matrix estimate} $\mathbf{Y} \in \mathbb{C}^{N \times N}$ identical to the definition in \eqref{eq:spectral} as:
\begin{equation}\label{eq:backproj}
\Y = \frac{1}{M} \mathcal{A}^H (\y),
\end{equation}
which is recognized as the backprojection estimate of the lifted ground truth $\x \x^H$.
%As a result, the leading eigenvector of the spectral matrix $\mathbf{Y}$ is approximately a minimizer of \eqref{eq:spec2} by definition. 
As a result, the classical spectral method constructs an initial estimate as an approximate minimizer of the least squares criterion over the feasible set of rank-1, PSD matrices. 
%A clear intuition of this generalization is one that is geometric, considering that in an ideal 

\subsubsection{Similarity of phaseless measurements}
The spectral method is inherently a search of the domain of rank-1, PSD matrices of the form $\z \z^H$ to attain a match in the range of $\mathcal{A}$, as:
\begin{align}
\v^H \mathbf{Y} \ \v =  \frac{1}{{M}}\sum_{m = 1}^{{M}} y_{m} \v^H \mathbf{a}_m \mathbf{a}_m^H \v  &= \frac{1}{{M}}\sum_{m = 1}^{{M}} y_{m} | \langle \mathbf{a}_m, \v \rangle |^2. 
% &= \frac{1}{M} \sum_{m = 1}^{{M}} y_{m} \hat{y}_m(\v). 
\end{align}
Inserting a proper scaling by an estimate $\lambda_0$ of the energy of the true solution $\x$ (i.e. $\| \x \|^2$), \eqref{eq:specmet} essentially corresponds to choosing an estimate that matches the phaseless measurements at hand the highest on average, i.e.,
\begin{equation}\label{eq:specmet2}
\hat{\x} = \underset{\|\v \| = \sqrt{\lambda_0}}{\text{argmax}} \ \frac{1}{M}  \langle \y, \hat{\y}(\v) \rangle,
\end{equation}
where $\hat{\y}(\v)_m =  | \langle \mathbf{a}_m, \v \rangle |^2$ with $\| \v \| = \sqrt{\lambda_0} \approx \| \x \|$. 
%Note that the spectral content of matrix $\mathbf{Y}$ is scale invariant, hence Euclidean-sense alignment in \eqref{eq:specmet2} is valid in the absence of $\lambda_0$. 

% and the absolute value in \eqref{eq:specmet} is dropped since $\y, \hat{\y}(\v) \in \mathbb{R}_{+}^M$ by definition. 

Observe that \eqref{eq:specmet2} outlines a fundamental issue with the classical spectral initialization. 
To generate an accurate estimator, the method relies on the signal domain correlation to be implied by a Euclidean sense alignment of the corresponding phaseless measurements on $\mathbb{C}^M$. % where a search is performed over rank-1, PSD matrices for a match in the range of the lifted forward model $\mathcal{A}$. 
However, \eqref{eq:specmet2} lacks a tractable normalization step to be a proper measure of the alignment of the phaseless measurements in the range of $\mathcal{A}$, unless the restricted-isometry type property in \eqref{eq:tightframe} is satisfied very tightly over all $\z \in \mathbb{C}^N$. 
Indeed such a concentration bound is not feasible for the Gaussian sampling model at information theoretic limits of $M = \mathcal{O}(N)$ \cite{chen2017solving, yonel2020deterministic}. 
%This is indeed not feasible for the Gaussian sampling model unless $M \geq \mathcal{O}(N \log N)$.  
Thereby, a large average correlation value could be obtained by an estimate that is biased towards a match with high magnitude elements in $\y$ as opposed to a uniform match over all samples which is ideal of an accurate estimator under the assumption that $\mathcal{A}$ is injective with $M \geq 4N-4$. 

Consequently, \eqref{eq:specmet2} reveals a key intuition to why \emph{sample truncation} improves spectral methods: the high magnitude measurements are excluded from the correlation measure in \eqref{eq:specmet2} 
%as:
%\begin{equation}
%\mathbf{Y} = \frac{1}{{| \mathit{I}_{\kappa}}| }\sum_{m \in \mathit{I}_{\kappa}} y_{m} \mathbf{a}_m \mathbf{a}_m^H,
%\end{equation}
which mitigates this effect resulting from variations on the length of elements that are mapped from the unit ball in $\mathbb{C}^N$ to the range of $\mathcal{A}$.
In addition, sample truncation directly relates to the feasibility of a uniformly tight \eqref{eq:tightframe} when $M = \mathcal{O}(N)$ in the Gaussian sampling model \cite{chen2017solving}, which directly implies that such variations are bounded with high probability at more practical sample complexities for high dimensional inference problems.

\subsection{Sample Processing and Geometric Formulations}\label{sec:OrthMinnrom}
%\subsubsection{Sample Pruning and Truncation}
In essence, the formation of $\mathbf{Y}$ via backprojection  in the lifted domain as \eqref{eq:backproj} can be interpreted as a linear synthesis in the range of the adjoint of the lifted forward model, i.e., in the span of rank-1, PSD measurement matrices $\{ \a_m \a_m^H \}_{m = 1}^M$.
This motivates deploying more general processing functions $\mathcal{T}(\cdot)$ on the phaseless measurements beyond sample truncation in order to perform the synthesis such that the principal spectral component of 
\begin{equation}\label{eq:preproc0}
\mathbf{Y}(\mathcal{T}) = \frac{1}{{M}}\sum_{m = 1}^{{M}} t_m \mathbf{a}_m \mathbf{a}_m^H, \quad \text{where} \ \ \mathbf{t} = \mathcal{T}(\y) \in \mathbb{C}^M
\end{equation}
better approximates the underlying ground truth $\x \in \mathbb{C}^N$.   
The view in \eqref{eq:preproc0} culminates into several alternative approaches for processing the phaseless measurements, and motivates using pre-processing functions $\mathcal{T}$ that mitigate shortcomings of the classical formulation.

The particular choice of $\mathcal{T}$ is a crucial design step for spectral methods. 
The general formalism in \eqref{eq:preproc0} includes truncation approaches, which were the first pre-processing techniques pursued for improving spectral methods.
%Truncation refers to setting $\mathcal{T}$ to a hard-thresholding function on $y_m$ such that only a subset of $m = 1, \cdots, M$ are included in the synthesis of the data matrix $\Y$.
Notable choices for $\mathcal{T}$ in the literature include truncation with respect to a hyper-parameter $\kappa$-standard deviations from the sample mean \cite{chen2017solving, wang2017solving} or the sample median \cite{zhang2016provable}, decaying exponential functions of measurements \cite{gao2017phaseless, ghods2018linear}, and using reduced orders of measurements as $y_m^{k}$ with $0< k <1$ in the synthesis \cite{wang2018phase}. 
More often than not, sample processing functions are formulated in the literature through statistical arguments or heuristics to improve the accuracy of estimation by mitigating the inherent biasing effect in the classical spectral method. 
To this end aforementioned sample processing functions aim to prune the weight of sampling vectors that correspond to high magnitude measurements in the synthesis of $\Y$. 
A handful of approaches adopt geometric derivations for spectral estimates of the form \eqref{eq:preproc0}.
Notable such methods in the literature are the linear spectral estimator of \cite{ghods2018linear} which uses \emph{the minimum norm solution} of \eqref{eq:objfun2} in constructing the spectral matrix estimate $\Y_{min}$,
%\footnote{We provide derivations and detail the equivalence in Appendix \ref{}.} $\Y_{min}$, 
and orthogonality promoting methods \cite{wang2017solving, wang2017scalable, chen2018phase, duchi2019solving} which leverage the near-orthogonality property of high dimensional random vectors.

In particular, orthogonality promoting methods set $\mathcal{T}$ as the indicator function of an index set that picks the subset of sampling vectors with $m \in \mathit{I}_{S}$, where $\mathit{I}_{S}$ is the index set consisting of the $S$ smallest magnitude measurements.
As a result the sampling vectors included in synthesis are uniformly scaled, and an initial estimate is formed using the eigenvector associated with the \emph{minimum} eigenvalue. 
Inherently, these methods seek a \emph{robust fit} using the $\ell_1$ loss over the small magnitude measurements in the synthesis under the generalized gradient estimator of \cite{bhojanapalli2016dropping}
%An effective generalization of spectral estimation to arbitrary \emph{convex} loss functions over the lifted variable is via the gradient of $\mathcal{J}(\X)$ as \cite{}:
\begin{equation}\label{eq:spect}
\mathbf{Y}(\mathcal{J}) := \mathcal{P}_{+} \left( - \nabla \mathcal{J} (\X) \vert_{\X = \mathbf{0}} \right), 
\end{equation}
where $\mathcal{J}$ is a convex function over the lifted domain, and $\mathcal{P}_{+}$ denotes the orthogonal projection onto the PSD cone in $\mathbb{C}^{N \times N}$. 
When $\mathcal{J}$ is set as the $\ell_2$ loss over the lifted variable, this initialization reduces to the classical spectral method. 

In principle, the generalized gradient estimator is motivated geometrically as it is equivalent to keeping tangent of $\mathcal{J}$ at $\X = \mathbf{0}$, evaluated over the set of rank-1, PSD matrices.
Indeed the spectral matrix estimate is analogously obtained from the first Uzawa iteration that solves the following problem over the lifted domain:
%Indeed, the estimator is analogously obtained from the first Uzawa iteration that solves the rank-1, PSD constrained problem over the lifted domain using $\mathcal{J}$ objective to
\begin{equation}\label{eq:Plift_cp4}
\text{minimize} \ \mathcal{J}(\X) \quad s.t. \quad \mathbb{C}^{N \times N} \ni \X \succeq 0
\end{equation}
where the optimization is initiated from $\X = \mathbf{0}$, with $\mathcal{J}$ evaluating the mismatch in the range of $\mathcal{A}$.
Specifying $\mathcal{J}(\X) = \frac{1}{M} \| \mathcal{A}_{\mathit{I}_{S}}(\X) - \y_{\mathit{I}_S} \|_1$ then yields the spectral matrix of the orthogonality promoting method, where%\footnote{See Appendix \ref{app:Chp5_l1} for the derivation of the $\ell_1$ spectral estimate.}
\begin{equation}\label{eq:l1Spec}
\mathbf{Y}_{\ell_1}(S) = \frac{1}{M} \sum_{m \in \mathit{I}_S} \a_m \a_m^H. 
\end{equation}
Hence, the orthogonality promoting formulation is based on the intuition that \eqref{eq:spect} should form an estimate that is aligned orthogonal to the rank-1, PSD unknown in the lifted domain when the loss is constructed as a robust fit to near-orthogonal measurements of the underlying ground truth signal. 
%Notably, this formulation is related to the optimizationless PhaseLift method of \cite{}, where it is shown that under identical conditions to the recovery guarantees in PhaseLift, solving \eqref{eq:} provides the exact solution to the phase retrieval problem. 

%The justification of this geometric estimator is that it minimizes the tangent of $\mathcal{J}(X)$ on $X = 0$ over the set of rank-1, PSD matrices.

%\begin{equation}\label{eq:spect}
%\mathbf{Y} := \mathcal{P}_{+} \left( - \nabla \mathcal{J} (\X) \vert_{\X = 0} \right), 
%\end{equation}
%where $\mathcal{P}_{+}$ denotes the orthogonal projection onto the PSD cone in $\mathbb{C}^{N \times N}$. 

%Fundamentally, a sufficiently large subset of measurements with $S \geq N$ are absolutely necessary for constructing a spectral matrix that has a trivial null-space to avoid ambiguous spectral estimators via the minimal eigenvalue. 

Notably, under the assumption that we are operating near the information limits, as $M \ll N^2$, an optimization problem of the form of \eqref{eq:Plift_cp4} is severely ill-posed. 
Indeed, the original formulation of the convex optimization problem in the lifted domain utilizes a low rank constraint via the trace norm of $\X$ to promote exact solutions through iterative solvers using the original quadratic loss objective in \eqref{eq:objfun2} for $\mathcal{J}$. 
To this end, one can in turn seek \emph{the minimum norm solution} in the initialization to promote a closed form estimate, which accordingly regularizes the problem over the lifted domain as:
\begin{equation}\label{eq:minnorm}
\text{minimize:} \ \| \X \|_F^2 \quad \text{s.t.} \quad \mathcal{A}(\X) = \y,
\end{equation}
in favor of the PSD projection in \eqref{eq:Plift_cp4}. 
Solution of \eqref{eq:minnorm} is equivalently the least squares solution to phaseless equations over the lifted domain via the right pseudo-inverse: 
\begin{equation}
\Y_{min} = \mathcal{A}^H (\mathcal{A} \mathcal{A}^H)^{-1} ( \y ),
\end{equation}
which is analogously a \emph{filtered-backprojection} estimate of the lifted unknown, where $\mathcal{A} \mathcal{A}^H = |\mathbf{A} \mathbf{A}^H|^2$, with the squared absolute value denoting \emph{element-wise} operations. %, the computational complexity of this operation does not increase total complexity of spectral initialization as long as $M \approx \mathcal{O}(N)$ as assumed. 
Clearly, $\Y_{min}$ yields an optimal estimator in terms of MSE over the measurement domain, as opposed to the original backprojection estimate.
However, this comes at the additional computational cost of solving an $M \times M$ system of equations for the filtering step. 

The minimum norm solution promotes an estimate by facilitating an alignment on the lifted domain in $\mathbb{C}^{N \times N}$. The spectral estimation is conducted by a search for the rank-1, PSD element that is maximally aligned to the orthogonal projection of the lifted ground truth $\x \x^H$ onto the range of $\mathcal{A}^H$ such that
\begin{align}
 \langle \z^H \Y_{min} \z \rangle  &= \langle 	\z \z^H, \mathcal{A}^H (\mathcal{A} \mathcal{A}^H)^{-1} ( \y )	\rangle_F \label{eq:minrmsol_e0} \\
&=   \langle \z \z^H,  \mathcal{A}^H (\mathcal{A} \mathcal{A}^H)^{-1} \mathcal{A}( \x \x^H )	\rangle_F \label{eq:minrmsol_e1} \\
&=   \langle \mathcal{A} (\z \z^H),  \mathcal{A}( \x \x^H )	\rangle_{(| \mathbf{A} \mathbf{A} |^{2})^{-1}} \label{eq:minrmsol_e2},
\end{align}
where $\mathcal{T}(\y) = (\mathcal{A} \mathcal{A}^H)^{-1} ( \y ) := (| \mathbf{A} \mathbf{A} |^{2})^{-1}  \y$ and $\y := | \mathbf{A} \x |^2$ with $\mathbf{A} \in \mathbb{C}^{M \times N}$. 
Equivalently through \eqref{eq:minrmsol_e2}, the estimator 
%Hence, the minimum norm spectral estimation 
seeks an alignment between phaseless measurements in the range of $\A$ that are \emph{whitened}, i.e., includes a correction for the correlations between the sampling vectors. 

Despite solid geometric foundations, there exists clear limitations for these approaches. 
Although an improved performance should be expected using the minimum norm solution over the classical formulation, this is attained at the cost of inverting an $M \times M$ system of equations. 
%Furthermore, the numerical demonstrations of \cite{ghods2018linear} indicate the method is most effective in low dimensional problems with high over-sampling ratios.
%This in fact is an intuitive outcome of the fact that an estimate can at best be pursued by alignment to the projection of $\x \x^H$ onto the range of $\mathcal{A}^H$, which is contingent on the descriptive capacity of the collection $\{ \a_m \a_m^H \}_{m = 1}^M$ as a dictionary for rank-1, PSD matrices in $\mathbb{C}^{N \times N}$.
Furthermore, as the estimator is limited to maximizing alignment to the projection of $\x \x^H$ onto the range of $\mathcal{A}^H$, its performance is highly contingent on the descriptive capacity of the collection $\{ \a_m \a_m^H \}_{m = 1}^M$ as a dictionary for rank-1, PSD matrices in $\mathbb{C}^{N \times N}$. 
Indeed, the numerical demonstrations of \cite{ghods2018linear} indicate the method is most effective in low dimensional problems with high over-sampling ratios. 
On the other hand while near-orthogonality assumptions are satisfied with high probability for the Gaussian sampling model, its a notable limitation for orthogonality promoting methods in problems with structured measurements. 
Particularly, the analysis in \cite{wang2017solving} requires the set of near-orthogonal sampling vectors to have cardinality of at least $3/4^{th}$ of total number of samples $M$, whereas empirically recommends $| \mathit{I}_{S}^c | = M/6$ when $M/N$ is near information theoretic limits.

%Furthermore, 
%the numerical demonstrations of \cite{ghods2018linear} indicate the method is most effective in low dimensional problems with high over-sampling ratios. 

\subsection{Optimal Sample Processing in the Gaussian Model}
More recently, for the Gaussian sampling model, optimal pre-processing functions have been studied in few seminal works \cite{lu2017phase, mondelli2017fundamental, luo2019optimal}. 
In \cite{lu2017phase}, authors identify the phase transitions of spectral methods in that the method fails with probability $1$ below a critical factor of samples with $M = \alpha N$ if $\alpha \leq \alpha^*$, in the limit $N \rightarrow \infty$. 
In \cite{mondelli2017fundamental}, again fixing $\alpha = M/N$ such that $N \rightarrow \infty$, authors define the \emph{weak recovery} guarantee in the noise-free setting as yielding a spectral estimate such that $\epsilon > 0$ is attained with probability $1$. 

Assuming that the underlying unknown of interest is unit norm, i.e. $\| \x \| = 1$, and the measurements are realized from Gaussian sampling vectors with unit variance, the following $\mathcal{T}(\cdot)$ is derived as optimal in the sense that weak recovery is guaranteed with smallest possible oversampling factor $\alpha \geq \alpha^*$:
\begin{equation}\label{eq:preproc}
\mathcal{T}(y_m) = \frac{y_m -1}{y_m + \sqrt{\alpha} - 1}.
\end{equation}
as $\sigma \rightarrow 0$ with $\sigma^2$ denoting the variance of additive white Gaussian noise on the phaseless measurements. 
%which approaches to $1 - 1/y_m$ as $\delta \rightarrow_{+} 1$, 
Additionally the impact of additive noise is captured in the lower bound $\delta^* = 1 + \sigma^2 + {o}(\sigma^2)$ as $\sigma \approx 0$. 
On the other hand in \cite{luo2019optimal}, the optimality of $\mathcal{T}$ is considered over all $\alpha$ by considering the maximization over the correlation coefficient $\rho$,    
%under the identical setting, yielding the limit $\delta \rightarrow_{+} 1$ of \eqref{eq:preproc} as optimal for the Gaussian sampling model with
yielding the limit $\delta \rightarrow_{+} 1$ of \eqref{eq:preproc} as optimal for the noise-free Gaussian sampling model with
\begin{equation}\label{eq:optpre}
\mathcal{T}^*(y_m) := \underset{\mathcal{T} \in \mathit{F}_{\mathcal{T}}}{\text{sup} }  \ \rho( \x, \alpha) \rightarrow  1 - \frac{1}{y_m},
\end{equation} 
where $\mathit{F}_{\mathcal{T}}$ denotes the set of feasible pre-processing functions that the authors define as having bounded range with an upper bound that is positive valued. 
Note that the optimal function in the noise-free Gaussian setting approaches to the right-hand-side of \eqref{eq:optpre}, which does not have a bounded range. 
{In addition, the sampling factor $\alpha$ now appears in the correlation coefficient in \eqref{eq:optpre} unlike the weak recovery framework of \cite{mondelli2017fundamental}, which considers the existence of an absolute $\epsilon > 0$ with an $\alpha$-adaptive processing function.}
In fact, the sample processing function in \eqref{eq:preproc} is merely a \emph{smoothed} version of the optimal function in \eqref{eq:optpre} as:
\begin{equation}
\mathcal{T}(y_m) = 1 - \frac{1}{(\sqrt{{1}/{\alpha}}) y_m  + (1 - \sqrt{{1}/{\alpha}})},
\end{equation}
where the denominator is effectively a \emph{convex combination} of each measurement $y_m$ with its \emph{expected value} in the Gaussian sampling model under the unit norm assumption on $\x$, as $\mathbb{E}[y_m] = \frac{1}{2} \| \x \|^2 \mathbb{E} [\chi_2^2] = 1$. 
On the other hand the framework in \cite{luo2019optimal} promotes pre-processing functions that depend on the noise level in the measurements, in which optimal processing functions were derived under the Gaussian sampling model with Poisson and additive white Gaussian noise. 

%Despite the specificity of the accompanying optimality results to the Gaussian sampling model, the 

Despite their accompanying performance guarantees, the aforementioned works are inherently limited by the model assumptions as they were shown as optimal only for measurements that were collected from i.i.d. Gaussian sampling vectors. 
Further restrictions on the optimality arguments include element-wise structural form of the pre-processing functions, and the search space constraints as depicted in \eqref{eq:optpre}. 
Nonetheless, there is a crucial insight to gain from the sample processing approaches to spectral initialization that is especially striking due to their reciprocal functional forms with respect to $\y_m$.
Following from our deconstruction that the spectral estimation is a search of alignment over phaseless measurements, the optimally designed sample processing functions indicate that there exists a mechanism beyond that of pruning the impact of sampling vectors at synthesis to mitigate the biasing effect under the statistical formalism of spectral methods.
In fact, the geometric intuition that is inherent in the classical spectral method and its heuristic variants is effectively eradicated by the sample processing functions that are optimally derived under the Gaussian sampling model, as there is no particular relation to an Euclidean-sense alignment under such non-linear transformations of the measurements. 
Therefore, it begs to question why such a reciprocal relationship to the value of $y_m$ could be optimal. 
To this end, in our unique perspective we interpret the role of the sample processing function $\mathcal{T}$ as \emph{warping} the $\ell_2$-sense alignment measure to metrics that are more suitable for assessing similarity between phaseless measurements. 
%Under such a formulation, the estimation  independently from any model specification. 

%is that sample processing functions have the potential to warp the $\ell_2$-sense alignment metric to more effective metrics of similarity between phaseless measurements. 
%To this end, we propose deriving spectral estimators through approximately minimizing Bregman distances over phaseless measurements, and formulate a novel framework that promotes designs for $\mathcal{T}$ independently from any model specification on the sampling vectors. 

\section{Bregman Divergence Minimizing Spectral Methods}\label{sec:4_sec3}

%In summary, the inherent mechanisms of spectral estimation are analyzed from a geometric perspective that is consistent across different methods in the literature. 
%\begin{itemize}
%\item Approximately minimize a loss function of choice over phaseless measurements in the range of $\mathcal{A}$,
%\item Tractably search for an estimate in the set of rank-1, PSD matrices through an eigen-decomposition.
%\end{itemize}

%\subsection{Motivation}
In this section, we formulate the design of spectral methods through principles of approximate Bregman distance minimization over phaseless measurements in the range of the lifted forward model. 
%To this end we proceed from our deconstruction of spectral estimation, and we build a novel framework that is model-free, and tractable via a spectral search over the set of rank-1, PSD matrices. 

%In fact, the $\ell_2$-loss that is approximately minimized in spectral estimation is an instance of a Bregman distance, albeit not one that is most suited for applications to phase retrieval. 
%Considering that the alignment of phaseless measurements is 

\subsection{Preliminaries on Bregman Loss Functions}\label{sec:Sec4_1}
We begin by presenting preliminary definitions and properties which we feature in deriving optimal sample processing functions via our framework. 
\begin{definition}{(\emph{Bregman divergence} \cite{banerjee2005clustering})} \label{def:Bregman}
Let $\phi: \mathit{S} \rightarrow \mathbb{R}$, be a strictly convex function defined on a convex set $\mathit{S} = \mathrm{dom}(\phi) \subset \mathbb{R}^M$ such that $\phi$ is differentiable on the relative interior $\mathrm{ri}(\mathit{S})$, assumed to be nonempty. The {Bregman divergence} $d_{\phi}: \mathit{S} \times \mathrm{ri}(\mathit{S}) \rightarrow [0, \infty)$ is defined as:
\begin{equation}\label{eq:Bregman}
d_{\phi} (\q , \p ) = \phi(\q) - \phi(\p) - \langle \q - \p, \nabla \phi(\p) \rangle,
\end{equation}
where $\nabla \phi(\p)$ is the gradient of $\phi$ evaluated at $\p$. 
\end{definition}
In other words, Bregman divergence is a measure of distance between points $\q, \p \in \mathit{S}$ defined in terms of a strictly convex function of choice in $\phi$. 
Geometrically, \eqref{eq:Bregman} characterizes the gap between a strictly convex function and its tangent at a given point $\p \in \mathrm{ri}(S)$.
As such, squared Euclidean distance is a particular instance of a Bregman divergence using the squared $\ell_2$ norm on $\mathbb{R}^M$ as:
\begin{equation}
d_{\ell_2} (\q , \p ) = \|\q \|^2_2 -  \|\p \|^2_2- \langle \q - \p,  2\p \rangle = \| \q - \p \|_2^2.
\end{equation}
Fundamentally, our motivation in pursuing general Bregman distances over phaseless measurements is to transcend the $\ell_2$-sense similarity that is relied upon in applications of spectral estimators for phase retrieval.
We stressed the major shortcoming of this inherent mechanism, where maximizing the correlation of phaseless measurements is prone to biasing a fit over a subset of measurements that correspond to high magnitudes due to the averaging effect. 

Ultimately, under the injectivity assumption on the lifted forward operator with $M \geq 4N - 4$, an accurate estimator ideally exhibits a uniform fit over measurements. 
In information theory, such a fit is pursued %between discrete probability distributions
via a Bregman distance generated over the more exclusive domain in $\mathbb{R}_{+}^{M}$ using the negative entropy function as \cite{banerjee2005clustering}:
\begin{align}\label{eq:KL-div}
d_{\phi}(\q , \p )  &= \sum_{m = 1}^M q_m \log q_m - \sum_{m = 1}^M p_m \log p_m \\
&-  \sum_{m = 1}^M ( q_m - p_m ) (\log p_m + 1) \nonumber \\
&= \sum_{m = 1}^M q_m \log \frac{q_m}{p_m} - \sum_{m = 1}^M ( q_m - p_m ),
\end{align}
which on the $M-\mathrm{Simplex}$ reduces to $D_{KL}(\q, \p)$, i.e., the KL-divergence of two discrete probability distributions having $\sum_{m} q_m = \sum_{m} p_m  = 1$. 
%Notably, the KL-divergence is a non-symmetric measure of distance between two distributions, hence there exists two alternative manners to pursue minimization with respect to a given distribution $\p$.
Beyond the $M-\mathrm{Simplex}$, the general form of \eqref{eq:KL-div} corresponds to the generalized $I$-divergence over $\mathbb{R}_{+}^M$.  

An equivalent interpretation of the generalized $I$-divergence in \eqref{eq:KL-div} is as weighted sum of the form:
\begin{equation}\label{eq:KL-div2}
D_I(\q , \p ) = (\sum_{j = 1}^M q_j) \sum_{m = 1}^M \frac{q_m}{\sum_{j = 1}^M q_j} \left(\frac{p_m}{q_m} - 1 - \log \frac{p_m}{q_m} \right),
\end{equation}
where a particular measure of mismatch over $\p$ and $\q$ in $\mathbb{R}_{+}^M$ is minimized in expectation over a proper $M-point$ distribution $\tilde{\q} = \q/\|\q\|_1$. 
%This particular measure is precisely the Itakura-Saito distance of two elements in $\mathbb{R}_{++}$, which is the Bregman distance generated by the function $\phi(x) = -\log x$ as
%\begin{equation}
%d_{\phi}(p_m, q_m) = \frac{p_m}{q_m} - 1 - \log \frac{p_m}{q_m}.
%\end{equation}
Clearly, without constraining the domain of $\q$ to the $M-\mathrm{Simplex}$, \eqref{eq:KL-div2} has a multiplicative bias on the scale of $\q$.
Instead we can evaluate the mismatch by an expectation over a distribution $\tilde{\q}$ of choice and remove the multiplicative bias such that
\begin{equation}
D_{\tilde{\q}}(\q , \p ) := \sum_{m = 1}^M \tilde{q}_m \left(\frac{p_m}{q_m} - 1 - \log \frac{p_m}{q_m} \right).
\end{equation}
We can interpret the distribution $\tilde{\q}$ as invoking prior information for the fit that is being pursued, where the Bregman distance is generated by a weighted function $\phi_{\tilde{q}}(\p) = -\sum_m \tilde{q}_m \log p_m$. 
Assuming a uniform prior as $\tilde{q}_m = 1/M$ yields the Itakura-Saito divergence from $\p$ to $\q \in \mathbb{R}_{++}^M$ as
\begin{equation}\label{eq:Itakura}
{d}_{\phi}(\p, \q) = \frac{1}{M}\sum_{m = 1}^M \left( \frac{p_m}{q_m} - 1 - \log \frac{p_m}{q_m} \right),
\end{equation}
where the generating function of the Bregman distance is the negative entropy rate. 

The uniform expectation in the Itakuro-Saito divergence is common practice for the signal processing setting, where indexes correspond to frequencies which are equally weighted is applications such as spectral density estimation where
\begin{align}\label{eq:ItSait}
D_{IS}(F, G) = 
%&:=\lim_{M \rightarrow \infty}  \frac{1}{M} \sum_{m = 1}^M \left( \frac{ |F(\mathrm{e}^{\mathrm{j}\omega_m})|^2}{\frac{|A_F(\mathrm{e}^{\mathrm{j}\omega_m})|^2}{\sigma_F^2} } - 1 - \log  \frac{|A_G(\mathrm{e}^{\mathrm{j}\omega_m})|^2}{| A_F(\mathrm{e}^{\mathrm{j}\omega_m})|^2}  \right) \\
\frac{1}{2\pi}\int_{-\pi}^{\pi} \left( \frac{F(\omega)}{G(\omega)} - 1 - \log  \frac{F(\omega)}{G(\omega)} \right) d\omega,
\end{align}  
where $F$ and $G$ are power spectra %where $\omega_m = -\pi + 2\pi (m-1)/(M-1)$, and $F$ and $G$ are power spectra 
over the normalized frequency variable $\omega \in [-\pi, \pi]$. % with $2\pi/M = \Delta \omega \rightarrow d\omega$ as $M \rightarrow \infty$. 
Notably, Itakura-Saito distance is equivalently the generalized $I$-divergence between the power spectral density of equal mean, stationary Gaussian processes \cite{banerjee2005clustering}.
In particular, for two zero mean Gaussian processes with spectral densities $F$ and $G$, consider the corresponding $M$-point probability densities $p_F^M(\X)$, $p_G^M(\X)$ where $\X = [ X_1, X_1, \cdots X_{M}]$.
The KL-divergence of the $M$-point probability distributions is defined as the $M^{th}$ order information criterion where \cite{kailath1967divergence, gray1980distortion}
\begin{align}\label{eq:infcri}
&I_M (F, G) := D_{KL}(p_F^M, p_G^M) \\
&= \int \cdots \int dX_1 \cdots dX_M \ p_F^M(\mathbf{X}) \log \frac{ p_F^M(\X) }{p_G^M (\X) }.
\end{align}
As noted in \cite{gray1980distortion}, it is known from Pinsker \cite{pinsker1964information} that the discrimination information satisfies:
\begin{equation}\label{eq:limcri}
I(F, G) := \lim_{M \rightarrow \infty} \frac{1}{M} I_M(F, G) = \frac{1}{2} D_{IS} (F, G).
\end{equation}
Hence, the Itakura-Saito divergence between the power spectral densities of two zero mean Gaussian processes is asymptotically equivalent to the KL-divergence of the distributions of the ensembles.
This naturally motivates the use of Itakura-Saito divergence as a loss function in problems with large $M$, where the asymptotic characterization in \eqref{eq:limcri} is sufficiently approximated in the finite sample regime. 
Most prominently, \eqref{eq:ItSait} and related techniques were utilized for speech processing and recognition problems \cite{itakura1975minimum, gray1981rate}.

% speech recognition systems \cite{}, and have been experimentally useful in pattern  non-Gaussian data 

%This motivates the use of Itakura-Saito divergence as a distortion measure in problems where $M$ is large

%This is especially utilized with auto-regressive models \cite{}, where an $M^{th}$ order auto-regressive model 

%This particular measure precisely equals the Itakura-Saito divergence of $\p$ and $\q$, which is a Bregman distance generated over the domain of $\mathbb{R}_{++}^M$ via the negative entropy rate function $\phi(\p) = -\sum_{m = 1}^M \log p_m$, as
%\begin{equation}
%d_{\phi}(\p, \q) = 
%\end{equation}

\subsection{Constructing Spectral Methods}
The key element in attaining spectral estimators for phase retrieval that minimize a choice of Bregman distance is the formulation of a quadratic form to maximize. 
Fundamentally, morphing the form of $d_{\phi}$ such that phaseless measurements are synthesized in the range of $\mathcal{A}$ allows for a tractable search over the domain of rank-1, PSD matrices.
However one must note that there exits an inevitable trade-off between exactness and tractability in minimizing Bregman distances by a spectral search due to the inherent non-linearity over the search variable $\q$.

To this end, a crucial property we leverage is the optimality of conditional mean as an estimator in minimizing the expectation of a Bregman divergence \cite{banerjee2005optimality}, which in effect is a generalization of minimum mean square error estimation. 
\begin{definition}{(\emph{Bregman Representative} \cite{banerjee2005optimality})}\label{def:BregRep}
Let $Q$ be a random variable that take values in $\mathit{X} \subset \mathit{S} \subset \mathbb{R}^M$ following a positive probability measure $v$ such that $\hat{\q} = \mathbb{E}_v [Q] \in \mathrm{ri}(S)$. Then for a Bregman divergence $d_{\phi}: \mathit{S} \times \mathrm{ri}(\mathit{S}) \rightarrow [0, \infty)$, $\hat{\q}$ is the unique minimizer of $\mathbb{E}_v [d_{\phi} (Q, \s)]$ over $\s \in \mathrm{ri}(\mathit{S})$, i.e.
\begin{equation}\label{eq:BregRep}
\hat{\q}:= \mathbb{E}_v [Q] = \underset{\s \in \mathrm{ri}(\mathit{S})}{\text{argmin}} \ \mathbb{E}_v [d_{\phi} (Q, \s)],
\end{equation}
where the minimum value is the \emph{Jensen gap} of $\phi: S \rightarrow \mathbb{R}$ defined in Definition \ref{def:Bregman}, such that
\begin{equation}\label{eq:JensenGap}
\mathbb{E}_v [d_{\phi} (Q, \hat{\q})] = \mathbb{E}_v [\phi(Q)] - \phi(\hat{\q}).
\end{equation}
\end{definition}
For the purpose of phase retrieval $\hat{\q}$ provides a surrogate for $\q$ in formulating a quadratic form that approximates minimization of \emph{any} Bregman distance from $\q$ to $\p$. 
Accordingly our framework for Bregman divergence minimization via a spectral search admits the following generic formalism:
\begin{align}
(P1)&: \underset{\q \in X \subset \mathit{S}}{\text{minimize}} \ d_{\phi}(\q, \p) \approx \ - \langle \q, h_{\phi}(\p, \hat{\q}) \rangle \label{eq:dmineq1}\\
(P2)  &: \underset{\| \v \| = 1}{\text{maximize}} \ \v^H \left[\sum_{m = 1}^M h_{\phi}(\p, \hat{\q})_m \a_m \a_m^H \right] \v, \label{eq:dmineq2}
\end{align}
where $X$ denotes the image of the set of rank-1, PSD matrices in the range of $\mathcal{A}$, with $h_{\phi}(\p, \hat{\q})$ serving the role of a sample processing function featuring the phaseless measurements \eqref{eq:phaless} within the definition of $\p$.
%The trade-off is then to minimally obtain the approximate form \eqref{eq:dmineq1} for a given Bregman divergence,  

Clearly, for such a formalism to be useful, obtaining the quadratic form in \eqref{eq:dmineq2} pertains conducting approximations in a manner that the maximization in $P2$ sufficiently reflects the minimization in $P1$. 
%Thereby, the definition of the Bregman representative provides 
%via the Bregman representative in \eqref{eq:dmineq1}
Therefore, our objective is to design a sample processing function $h_{\phi}$ such that the approximation in \eqref{eq:dmineq1} has \emph{minimal distortion} from the original Bregman divergence, while controlling the use of the Bregman representative to a minimum. % so that the mismatch of the two elements evaluated by the original divergence is maximally encoded by $P2$. 
%Such a principle in formulating $P2$ provides 
Intuitively, such a principled design of $h_{\phi}$ corresponds to maximally \emph{encoding} the mismatch information that is minimized in $P1$ into the formulation of $P2$ which is tractable as a spectral method. 

For conceptual brevity, we denote this as \emph{the principle of minimal Bregman representation}, which we mathematically characterize as follows.
\begin{definition}{\emph{Minimal Bregman Representation}.}\label{def:MinRep}
Given Definition \ref{def:BregRep}, let the Bregman divergence of choice in $d_{\phi}$ be approximated as \eqref{eq:dmineq1} via the Bregman representative of the random variable $Q$ over $\mathrm{ri}(S)$. Then, the formulation of $P2$ in \eqref{eq:dmineq2} introduces minimal distortion to the original problem in $P1$ if:  
\begin{equation}\label{eq:minrep}
\big| \mathbb{E}_v[-\langle Q, h_{\phi} (\p, \hat{\q}) \rangle + c_{\phi}(\p)] - \mathbb{E}_v [ d_{\phi} (Q, \p) ] \big| = \mathbb{E}_v [d_{\phi}(Q, \hat{\q})]
\end{equation}
where $c_{\phi}$ is a deterministic function with dependence on only the fixed quantity $\p \in  \mathrm{ri}(S)$. 
Furthermore, if an $h_{\phi}$ that satisfies \eqref{eq:minrep} is obtained under the least feasible order of approximations on $\q$, then $P2$ admits a minimal Bregman representation. 
\end{definition}
Subsequently, a sample processing function is referred to as \emph{optimal} under the following definition.
\begin{definition}{\emph{Optimal Sample Processing}.}\label{def:OptSP}
Given the Bregman divergence $\d_{\phi}: S \times \mathrm{ri}(S) \rightarrow \mathbb{R}$, let a spectral method in the form of $P2$ be constructed.
Then, if $h_{\phi}$ is obtained under a minimal Bregman representation, it is an optimal sample processing function for minimizing $\d_{\phi}$. 
\end{definition}

Essentially, if $h_{\phi}$ satisfies \eqref{eq:minrep} in approximating the original divergence, then $P2$ has minimal distortion from $P1$ as desired via the definition of the Bregman representative $\hat{\q}$. %, hence it is optimal for minimizing $d_{\phi}$. 
Consequently, once equipped with the Bregman representative of the image of any rank-1, PSD element in the range of the lifted forward model $\mathcal{A}$, we have the means to formulate the minimization of our Bregman divergence of choice into \eqref{eq:dmineq1} under the guidance of \eqref{eq:minrep}. 

As an immediate demonstration of our construction of $P2$, we address the formulation of the quadratic loss minimizing spectral method setting $\phi(\cdot) = \| \cdot \|^2$.
For the phase retrieval problem, the Bregman divergence between $\q$ and $\p$ then represents the $\ell_2$-mismatch between the rank-1, PSD elements in the range of the lifted forward model as depicted in \eqref{eq:objfun2}.
Having $\q = \mathcal{A}(\v \v^H)$ with $\| \v \| = 1$, the classical spectral method can be interpreted within our framework by setting $\p = \y$ and approximation in \eqref{eq:dmineq1} is conducted as:
\begin{equation}\label{eq:classicapprx}
d_{\ell_2}(\q, \p)\vert_{\p = \y} \approx (\| \hat{\q} \|^2 - \| \q \|^2) + \| \q \|^2 + \| \y \|^2 - 2\langle \q, \y \rangle.
\end{equation}
which fully eliminates the $\| \q \|^2$ dependence from $\d_{\phi}$ using the Bregman representative that minimizes $\mathbb{E}_v[\| \hat{\q} \|^2 - \| Q \|^2]$, and yields the maximization in $P2$ with $h_{\phi}(\p, \hat{\q})_m = y_m$. 

Notably, the classical spectral method does not require a normalization since the spectral search obtained for \eqref{eq:dmineq2} is invariant to the scaling on $\y$. 
However, despite introducing the minimal distortion term from the right-hand-side of \eqref{eq:minrep} to \eqref{eq:classicapprx} on expectation, the classical spectral method does not abide by the principle of minimal Bregman representation. 
Observe that the following approximation captures the dependence of $d_{\ell_2}(\q, \p)$ on $\q$ to full extent while attaining the form in \eqref{eq:minrep}:
\begin{align}
d_{\ell_2}(\q, \p) &= \langle \q, \q \rangle + \| \p \|^2 - 2\langle \q, \p \rangle =  \| \p \|^2 + \langle \q , \q - 2\p \rangle \\
& \approx  \| \p \|^2 + \langle \q , \hat{\q} - 2\p \rangle = c_{\phi}(\p) -\langle \q, h_{\phi} (\p, \hat{\q}) \rangle, \label{eq:specl2} 
\end{align} 
which on expectation satisfies:
\begin{equation}\label{eq:l2distor}
\mathbb{E}_v[\langle Q , \hat{\q} - 2\p \rangle]] = \mathbb{E}_v[\langle Q , Q - 2\p \rangle] + \mathbb{E}_v[\langle Q , \hat{\q} - Q\rangle]
\end{equation}
where $\mathbb{E}_v[\langle Q , \hat{\q} - Q\rangle] = \mathbb{E}_v [\phi(Q)] - \phi(\hat{\q})$ is by definition the minimal distortion identified as the Jensen gap in \eqref{eq:JensenGap}.
As a result, having discarded only $c(\p) = \| \p \|^2$, the spectral estimator in $P2$ admits a sample processing function as $h_{\phi}(\p, \hat{\q})_m = 2 p_m - \hat{q}_m$ when formulated under the principle of minimal Bregman representation.  

Here, we note that in addition to the violation of \eqref{eq:minrep}, \eqref{eq:specl2} also reveals that the Bregman representation is featured at a non-minimal order within the classical formulation. 
This is indeed a special case we observe with the quadratic loss, where the minimal distortion formulation of \eqref{eq:minrep} is sufficient for a minimal Bregman representation. 
Such a property is shared by the \emph{Mahalobinis} distance, where squared loss is defined as a quadratic form over a positive semi-definite matrix $\mathbf{T}$ as:
\begin{equation}
d_{\phi}(\q, \p) := \langle \q - \p, \q - \p \rangle_{\mathbf{T}} = (\q - \p)^H \mathbf{T} (\q - \p).
\end{equation}
Considering real valued $\q, \p$ for the purpose of phase retrieval, and assuming a real-valued $\mathbf{T}$ we have $\q^T \mathbf{T} \p = \p^T \mathbf{T} \q$ which yields
\begin{align}
d_{\phi}(\q, \p) &= \q^T \mathbf{T} \q - 2 \q^T \mathbf{T} \p + \p^T \mathbf{T} \p \\
&= - \langle \q, 2 \mathbf{T} \p - \mathbf{T} \q \rangle + c_{\phi}(\p) \\
&\approx - \langle \q, \mathbf{T} (2\p - \hat{\q}) \rangle + c_{\phi}(\p),
\end{align}
where we obtain a minimal Bregman representation akin to that in \eqref{eq:l2distor}, where the inner products are now defined under $\mathbf{T}$. 
Setting $\mathbf{T} = (| \mathbf{A} \mathbf{A}^H |^2)^{-1} := (\mathcal{A}\mathcal{A}^H)^{-1}$ then yields an analogous form to the minimum norm solution in \eqref{eq:minnorm}, where $\q = \mathcal{A}(\z \z^H)$, and $\p = \y = \mathcal{A}( \x \x^H)$.

%However for many practical applications of phase retrieval, a priori estimate on the norm of the underlying signal may not be feasible. 
%This motivates  

\subsection{Bregman Representatives for Phase Retrieval}
%\subsubsection{Bregman representations for phase retrieval}
A ramification of the formulation in \eqref{eq:specl2} is that %the inner product in \eqref{eq:specl2} is no longer scale invariant over $\p$.
the sample processing function derived with $h_{\phi}(\p, \hat{\q})$ no longer yields a spectral method that is scale invariant over $\p$ in $P2$.
Consequently, $P2$ requires a normalization step on the measurements $\y$ in the assignment of $\p$, such that it matches the scaling of the Bregman representative $\hat{\q}$ as
\begin{equation}\label{eq:l2min}
\hat{\x}_{\ell_2} := \underset{\| \v \| = 1}{\text{argmax}} \ \v^H \left[\sum_{m = 1}^M (\frac{2}{\lambda_0} y_m - \hat{q}_m )\a_m \a_m^H \right] \v,
\end{equation}
where $\lambda_0$ denotes an appropriate normalization factor on $\y$ that is available \emph{a priori} to spectral estimation.
Accordingly, $\lambda_0$ depends on the particular parameterization of the random variable $Q$ that generates $\q$ through the choice of the Bregman representative $\hat{\q}$.  

It may be recognized that we abuse the notation with $\lambda_0$ which earlier in Section \ref{sec:4_sec1} is used in denoting an estimate of the energy of the ground truth signal, i.e. $\| \x \|^2$.   
Considering that $\q$ is intrinsically generated by a spectral search over the unit sphere on $\mathbb{C}^N$ with $\| \v \| = 1$, an ideal normalization of $\y$ is indeed via the energy of $\x$ such that $\p$ corresponds to a phaseless measurement generated by an element in the same domain. 
This normalization is a theme inherently common to the optimal sample processing functions in the literature \cite{mondelli2017fundamental, luo2019optimal}, where the methods are derived under the assumption of a unit-norm unknown $\x$ for the Gaussian sampling model. 

In particular, an estimator for $\| \x \|^2$ is feasible in the case of high dimensional inference under the Gaussian sampling model, due to the near-isometry of the sampling vectors, which promotes the sample mean of $\y$ as a tight estimator for $\| \x \|^2$ with high probability \cite{candes2015phase}, i.e., $\lambda_0 = \| \y \|_1 / M$. 
Correspondingly, having $\a_m$ as i.i.d. random vectors following the complex Gaussian model, i.e. $\a_m \sim \mathcal{CN}(0, 1) := \mathcal{N}(0, \mathbf{I}/2) + \mathrm{j} \ \mathcal{N}(0, \mathbf{I}/2)$, the random vector $Q$ is defined with i.i.d. chi-square distributed entries with $2$ degrees of freedom, with a scaling of $\| \v \|^2/2$.
Then, given any fixed $\v$ on the unit sphere $\mathit{S}^{N-1}(1)$ in $\mathbb{C}^N$, the Bregman representative is evaluated with an expectation over the statistical model distribution as
\begin{equation}
\hat{\q}_m^{(1)} = \mathbb{E}_{\a_m \sim \mathcal{CN}(0, 1)} [ | \langle \a_m, \v \rangle|^2 ] = 1,
\end{equation}
which completes the spectral method in \eqref{eq:l2min}. 

In practical applications of phase retrieval such as those in imaging, measurements are structured either through physical modeling or sensing prior to inference.
To this end, few alternatives can be considered as the Bregman representative of $Q$ in the context of phase retrieval. 
A surrogate that is much more suitable is one that is evaluated with an expectation on $\v$ over the unit sphere given the sampling vectors $\a_m$.
Considering that the $\v$ generates the $m^{th}$ sample of the random variable $Q$ given the specific instance of $\a_m$, we assume a uniform distribution over the unit sphere in $\mathbb{C}^N$ as $\v \sim \mathrm{Unif}(\mathit{S}^{N-1}(1))$ which yields the following Bregman representative: %where % $\mathit{S}^{N-1}(1) \subset \mathit{B}^{N}(1)$ we obtain
\begin{equation}\label{eq:rep2}
\hat{\q}_m^{(2)} = \mathbb{E}_{\v \sim \mathrm{U}(\mathit{S}^{N-1}(1))} [ | \langle \a_m, \v \rangle|^2] = \frac{\| \a_m \|^2}{N}.
\end{equation}
Observe for the Gaussian sampling model that $\hat{\q}^{(2)}$ statistically tends to $\hat{\q}^{(1)}$ in high dimensions via strong law of large numbers as $\mathbb{E}_{\a_m \sim \mathcal{CN}(0,1)}[\hat{\q}_m^{(2)}] = \hat{\q}_m^{(1)}$. 
Thereby $\hat{\q}_m^{(2)}$ serves as a correction of $\hat{\q}_m^{(1)}$ for the finite sample setting.
%More critically, $\hat{\q}^{(2)}$ is an applicable representation for structured measurement models. 
Rather more critically, $\hat{\q}^{(2)}$ is a model independent Bregman representation, as the expectation is over the unit sphere given a collection of measurement vectors $\{\a_m \}_{m = 1}^M$. % which correspond to structured measurement models for practical purposes.

Yet, for the purpose of formulating $P2$ with structured measurements, $\hat{\q}^{(2)}$ proves to be insufficient on its own due to the normalization step required on the phaseless measurements. 
Namely, estimation of an appropriate normalization for $\y$ using the Bregman representative hinges on having a near-isometry property on the sampling vectors, as its operated under the assumption that the estimated scale accurately generalizes over the whole unit sphere. 
However, this is well known to be quite a stringent requirement for structured measurements that are governed by deterministic models, and therefore is the key motivation of minimizing Bregman divergences defined over the $M-\mathrm{Simplex}$. %, where elements are \emph{absolutely normalized}. 
With the specification of $M-\mathrm{Simplex}$ for the divergence domain, the random vector $Q$ is by definition normalized to sum to $1$ in order to in $S$, for which we obtain the Bregman representative as
\begin{align}\label{eq:rep3}
\hat{\q}_m^{(3)} &= \mathbb{E}_{\v \sim \mathrm{U}(\mathit{S}^{N-1}(1))} \left[ \frac{ | \langle \a_m, \v \rangle|^2}{\sum_{m =1}^M | \langle \a_m, \v \rangle|^2} \right] = \frac{\| \a_m \|^2}{\sum_{m =1}^M \| \a_m \|^2},
\end{align}
where we use the symmetry property over each dimension within the random variable in $Q$ indexed by $m$.

%We provide the details on derivation 

\section{Optimal Design of Sample Processing Functions}\label{sec:4_sec4}

In this section, we use our formalism outlined in Section \ref{sec:4_sec5}, and derive sample processing functions that facilitate the minimization of Kullback-Leibler, and Itakura-Saito divergences on the $M-\mathrm{Simplex}$ under the principle of minimal Bregman representation. 
%Namely KL-divergence, and the Itakura-Saito divergence.
%Namely, we refer to a sample processing function as optimal for minimizing the Bregman divergence of choice under the following definition.

Fundamentally, we stated the objective of our framework as constructing spectral methods that are provably good under universal arguments, and thereby improving their applicability in problems with structured measurements. 
% that are applicable to structured, deterministic measurement models for practical purposes of phase retrieval.
For many practical applications of phase retrieval, the underlying linear process is either sensed, or modeled as a deterministic mapping that captures physical phenomenon such as scattering, and do not attain favorable properties associated with the Gaussian sampling model. 
One such property we stressed relates to the proper normalization of phaseless measurements as highlighted in conclusion of Section \ref{sec:4_sec3}.
%which is inherently present in existing optimality results for the Gaussian sampling model. 

To this end, minimizing Bregman divergences on the $M-\mathrm{Simplex}$ presents a feasible option for practical purposes as the mismatch between two elements $\q$, $\p$ is evaluated under absolute normalization, rather than relatively to a quantity that is unknown at the receive end of a physical sensing system.

\subsection{Kullback Leibler Divergence Minimization}

Restricting the \eqref{eq:KL-div} on the $M-\mathrm{Simplex}$, the KL-divergence of $\q$ from $\p$ is defined as
\begin{equation}\label{eq:KLspec}
D_{KL}(\q, \p) = \sum_{m = 1}^M q_m \log \frac{q_m}{p_m}.
\end{equation}
For spectral estimation, the emphasis of our framework is on tractability in the form of $P2$ under the principle of minimal Bregman representation on the original divergence minimization problem in $P1$.  
The KL-divergence minimizing spectral method is thereby formulated setting $\p = \y / \| \y \|_1$, and using the Bregman representative $\hat{\q}^{(3)}$ within the logarithmic term. % such that
\begin{align}
D_{KL}(\q, \p) \approx - \sum_{m = 1}& q_m \log \left( \frac{y_m / \| \y \|_1}{\| \a_m \|^2 / \sum_{i = 1}^M \| \a_i \|^2} \right) \label{eq:KLaprx}\\
\hat{\x}_{KL} := \underset{\| \v \| = 1}{\text{argmax}} \ \v^H &\left[\sum_{m = 1}^M \log \left( \frac{y_m \sum_{i = 1}^M \| \a_i \|^2}{\| \a_m \|^2 \| \y \|_1}  \right) \a_m \a_m^H \right] \v.\label{eq:KLfinal}
\end{align}
\begin{proposition}
Let $P2$ in \eqref{eq:dmineq2} be constructed as the maximization in \eqref{eq:KLfinal} with $\p = \y/\| \y \|_1$ such that $h_{\phi}(\p, \hat{\q})$ denotes the sample processing function as:
\begin{equation}
[h_{\phi}(\p, \hat{\q})]_m := [\mathcal{T}(\y)]_m = \log \left( \frac{y_m \sum_{i = 1}^M \| \a_i \|^2}{\| \a_m \|^2 \| \y \|_1}  \right).
\end{equation}
Then, \eqref{eq:KLfinal} is a spectral method under minimal Bregman representation such that \eqref{eq:KLaprx} is on expectation the best approximation for the original KL-divergence in \eqref{eq:KLspec} with
\begin{equation}\label{eq:KL_justif}
\big| \mathbb{E}_{\v} [- \langle Q,  h_{\phi}(\p, \hat{\q}) \rangle] - \mathbb{E}_{\v} [D_{KL} (Q, \p)] \big| = \mathbb{E}_{\v}[D_{KL}(Q, \hat{\q})].
\end{equation}
\end{proposition}
%the spectral method is obtained under a minimal Bregman representation with $h_{\phi}(\p, \hat{\q})_m = \log (p_m / \hat{q}^{(3)}_m)$.
\begin{proof}
Particularly,
\begin{align}
- \langle \q,  h_{\phi}(\p, \hat{\q}) \rangle = &\sum_{m = 1}^M q_m \left(  \log \frac{\hat{q}_m}{p_m} - \log \frac{q_m}{p_m} + \log \frac{q_m}{p_m} \right) \\ 
=&  \sum_{m = 1}^M q_m \log \frac{\hat{q}_m}{q_m} + \sum_{m = 1}^M q_m \log \frac{q_m}{p_m}.
\end{align} 
Taking the expectation of both sides we have
\begin{equation}\label{eq:KL_justif}
\mathbb{E}_{\v} [- \langle Q,  h_{\phi}(\p, \hat{\q})] = -\mathbb{E}_{\v}[D_{KL}(Q, \hat{\q})] + \mathbb{E}_{\v} [D_{KL} (Q, \p)],
\end{equation}
where by definition of the Bregman representative the distortion term $\mathbb{E}_{\v}[D_{KL}(Q, \hat{\q})]$ is minimized with $\hat{\q} = \hat{\q}^{(3)}$ since $Q$ lies in the $M-\mathrm{Simplex}$. Then, the principle of minimal Bregman representation is hence satisfied, having $|\mathbb{E}_{\v}[D_{KL}(Q, \hat{\q}^{(3)})] | = \mathbb{E}_{\v}[D_{KL}(Q, \hat{\q}^{(3)})]$ from the non-negativity of the Bregman divergence.
\end{proof}
%As a result, $h_{\phi}(\p, \hat{\q}^{(3)})$ of \eqref{eq:KLfinal} is the optimal sample processing function for minimizing the KL-divergence in \eqref{eq:KLspec}.

As a result, $h_{\phi}(\p, \hat{\q}^{(3)})$ of \eqref{eq:KLfinal} is the optimal sample processing function for minimizing the KL-divergence in \eqref{eq:KLspec} on the $M-\mathrm{Simplex}$ by Definition \ref{def:OptSP}.
However, it should be noted that even though the desired form in $P2$ is obtained with minimal distortion in \eqref{eq:KLaprx}, the spectral search of \eqref{eq:KLfinal} effectively perturbs the domain of the minimization, as $\v \v^H$ is not necessarily mapped onto the $M-\mathrm{Simplex}$ by $\mathcal{A}$ in synthesizing $\q$. 
This is a fundamental limitation that arises in spectral methods, since normalization is not feasible in the range of $\mathcal{A}$. % the domain of rank-1, PSD matrices.  
In consequence, the optimal spectral method of \eqref{eq:KLfinal} for minimizing \eqref{eq:KLspec} is equivalent to the following optimization problem:
\begin{equation}\label{eq:KL_altform}
\underset{\q \in X \subset \mathit{S}}{\text{minimize}} \ \| \q \|_1 \left( D_{KL}(\frac{\q}{\| \q \|_1}, \frac{\y}{\| \y \|_1}) - D_{KL}(\frac{\q}{\| \q \|_1}, \hat{\q}^{(3)}) \right).
\end{equation}

In the light of \eqref{eq:KL_altform}, a useful interpretation for minimizing KL-divergence by \eqref{eq:KLfinal} is through \emph{variational} principles.
In machine learning and probabilistic models, variational inference refers to the minimization of \eqref{eq:KLspec} when it is intractable due to the normalization of the reference distribution $\p$, by using a tractable objective that squeezes the true divergence towards zero \cite{wainwright2008graphical}.
For the purpose of constructing spectral methods, the intractability that must be alleviated is rather rooted in the exhaustive search over the set of rank-1, PSD matrices.   
Accordingly, under the principle minimal Bregman representation, we introduce two sources of distortion on the original objective such that: 
\begin{itemize}
\item The KL-divergence to sum-to-$1$ normalized phaseless measurements can at best be minimized \emph{relative} to the KL-divergence to the Bregman representative. 

\item The relative minimization is tractable only upto a scaling factor influenced by the spectrum of $\mathbf{A}^H \mathbf{A} \in \mathbb{C}^{N \times N}$, as $\| \q \|_1 = \| \mathbf{A} \v \|_2^2$ with $\| \v \| = 1$. 
\end{itemize}

Observe that the method preserves the original objective to drive KL-divergence from $\y/\| \y\|_1$ towards $0$, with a regularizer that incentives maximizing the divergence from the Bregman representative. 
Since $\| \q \|_1 > 0$, \eqref{eq:KL_altform} is guaranteed to search a minimizer that is closer to $\y$ on the $M-\mathrm{Simplex}$ than to $\hat{\q}^{(3)}$ in KL-sense.
Hence, although the $\| \q \|_1$ multiplier may bias the estimator to align towards principle subspaces of $\mathbf{A}^H \mathbf{A}$, \eqref{eq:KL_altform} strictly constrains its influence such that the relative divergence is preserved as negative. 
To this end, a key property for the quality of estimation with the spectral method in \eqref{eq:KLfinal} is having phaseless measurements $\y$ that are sufficiently distinct from the Bregman representative such that the minimization is dominated the relative KL-divergence term rather than the spectral content of the normal operator $\mathbf{A}^H \mathbf{A}$. 
Furthermore, the impact of the scaling vanishes as the measurement model is well-conditioned, since the condition number of $\mathbf{A}^H \mathbf{A}$ quantifies the maximum scaling factor that can influence the minimization.
The objective \eqref{eq:KLfinal} then approaches to the relative KL-divergence with $\kappa(\mathbf{A}^H \mathbf{A}) := \mathrm{sup}_{\| \v \| = 1} \| \mathbf{A} \v \|_2^2 / \mathrm{inf}_{\| \v \| = 1} \| \mathbf{A} \v \|_2^2 \rightarrow 1$.

%To this end the condition number $\kappa(\mathbf{A}^H \mathbf{A})$ provides a worst-case it quantifies the maximum scaling factor that can influence the minimization. %hat can influence the minimization. 

%Essentially, the condition number of the normal operator $\mathbf{A}^H \mathbf{A}$ controls the worst case in \eqref{eq:KL_altform}, where $\kappa(\mathbf{A}^H \mathbf{A}) \rightarrow 1$ guarantees the 
%The impact of the spectrum of the normal operator $\mathbf{A}^H \mathbf{A}$ can be examined via its condition number $\kappa$.  

%In other words, as the condition number $\kappa(\mathbf{A}^H \mathbf{A}) \rightarrow 1$, the more the original KL-divergence minimization objective is captured in the spectral method. 

%On the other hand, a critical advantage of minimizing KL-divergence over the quadratic loss is the scaling with respect to the $\ell_1$ norm in the range of $\mathcal{A}$, as opposed to the $\ell_2$-norm. 
%This allows for controlling the a well-conditioned operator 

It is well know that KL-divergence is non-symmetric in its arguments, and the minimization with respect to one argument versus the other yields different properties for the optimization. 
Indeed, minimizing KL divergence with respect to the first variable is known to be \emph{mode seeking} whereas minimization with respect to the second variable is known as \emph{mean seeking} due to the aggressive penalization of $q_m = 0$ in the latter when $p_m \neq 0$. 
This penalization promotes a more inclusive match over event set via the mean-seeking KL, which raises the question why to choose the KL-divergence in \eqref{eq:KLspec} to minimize over $\q$ given our specifications for applications to phase retrieval. 
Although it is ideally preferable to use the inclusive KL divergence knowing that $\p$ already lies in the domain of our search in the range of the lifted forward model,
tractably minimizing: 
%In the light of this, observe that the mean-seeking divergence $KL(\p, \q)$ admits a form where $\q$ solely appears within the logarithm, such that
\begin{equation}\label{eq:KLforw}
D_{KL}(\p, \q) = \sum_{m = 1}^M p_m \log \frac{p_m}{q_m} = \sum_{m = 1}^M q_m \left( \frac{p_m}{q_m} \log \frac{p_m}{q_m} \right),
\end{equation}
by a spectral method of the form of $P2$ is clearly not feasible as $\q$ solely appears within a non-linear logarithmic term. 
Simply pursuing the factoring in \eqref{eq:KLforw} to formulate $P2$ using the Bregman representative introduces an arbitrary distortion of $D_{KL}(\p, \hat{\q}) - \mathbb{E}_{v}[D_{KL} (\p, Q)])$ that is not minimal under Definition \ref{def:BregRep}.
As a result, \eqref{eq:KLforw} does not admit a form where Bregman representation is a justifiable approximation in formulating $P2$ as it is for the mode-seeking KL through \eqref{eq:KL_justif}.

%Given the measurement vectors $\{ \a_m \}_{m = 1}^M$, and using a distribution over the unit sphere for $\v$, we characterize any $\q$ as a sample from the random variable $Q$ such that a tractable surrogate is obtained using the expected value of the intractable term within \eqref{eq:KL-div} as:
%\begin{equation}\label{eq:KL-surr}
%\tilde{D}_{KL} (\q, \p ) := \langle \q, \mathbb{E}_{\v} \left[ \log \frac{Q}{\p} \right] \rangle,
%\end{equation}
%which emulates the form of $h_{\phi}$ in \eqref{eq:dmineq1} as desired for formulating $P2$.
%Using Jensen's inequality we conclude that via the Bregman representative of $Q$ over the $M-\mathrm{Simplex}$, a majorization/minimization principle is utilized to minimize \eqref{eq:KL-surr} where
%\begin{equation}
%\tilde{D}_{KL} (\q, \p ) \leq \sum_{m = 1}^M q_m  \log \frac{\mathbb{E}_{\v}[Q]_m}{p_m} = \sum_{m = 1}^M q_m  \log \frac{\hat{q}^{(3)}_m}{p_m}.
%\end{equation}
%The variational perspective further demonstrates the incompatibility of the mean-seeking KL for the purpose of spectral estimation as the function in \eqref{eq:KLforw} is not convex or concave over $q_m \in \mathbb{R}_+$, which incapacitates the use of Jensen's inequality.  

%Evidently mean-seeking KL divergence under our framework for spectral estimation leads to a larger distortion for the minimization than in the mode-seeking KL divergence.  

%It is also clear that there 

\subsection{Itakura-Saito Divergence Minimization}

%With the equivalent form in \eqref{eq:KL_altform} we stress the existence of a multiplicative term that influences the spectral method in \eqref{eq:KLfinal} for minimizing the desired KL-divergence. 
Itakura-Saito divergence can be considered as a metric that mitigates a weighting by $\q$ within the definition of the generalized $I$-divergence, as we derive it using an expectation with the uniform distribution.
A feature that further motivates the use of Itakura-Saito divergence is demonstrated through our derivation in Section \ref{sec:Sec4_1}. 
Namely for $\q, \p \in \mathbb{R}^M_{++}$, $D_{IS}(\q, \p)$ relates to the generalized $I$-divergence in the \emph{reverse sense}, i.e., it is derived from $D_I(\p, \q)$ where we replaced a factor of the first argument with the uniform distribution. 
This promotes Itakura-Saito as an alternative for minimizing the mean-seeking divergence in \eqref{eq:KLforw} using our framework. %only first order Bregman representations are deployed as:

In particular, we aim to utilize the principle of minimal Bregman representation to minimize:
\begin{align}\label{eq:ISdiv}
D_{IS}(\q, \p)& =  \frac{1}{M}\sum_{m = 1}^M \left( \frac{q_m}{p_m} - 1 - \log \frac{q_m}{p_m} \right)\\
& =  \frac{1}{M}\sum_{m = 1}^M q_m \left( \frac{1}{p_m} - \frac{1 + \gamma}{q_m}  \right),
\end{align}
where the logarithmic term is denoted with 
\begin{equation}
\gamma := \frac{1}{M}\sum_{m = 1}^M \log \frac{q_m}{p_m} = \phi(\p) - \phi(\q)
\end{equation} 
having $\phi$ as the original potential function of the Itakura-Saito divergence by definition. 

In relevance to our framework, the Itakura-Saito divergence minimizing formulation has a clear distinction from its previously studied counterparts. 
%s a favorable distinction from the mean-seeking KL (or generalized $I$-divergence) counterpart highlighted in \eqref{eq:KLforw}. 
That is, there is no single approximation that facilitates the formulation of $P2$ using the Bregman representative.
Note that an absolutely minimal use of the Bregman representative comes at the cost of fully eliminating the logarithmic component and the DC term in $-1$, thereby discarding components integral to the mismatch measure.   
To address this issue, we consider a hierarchical approach to approximating the original divergence under a minimal Bregman representation through \eqref{eq:minrep}. 
 
We first observe that $D_{IS}$ admits a point estimate for the logarithmic term through the Bregman representative such that 
%a point estimate for $\gamma$ can be used within the divergence rather than an element-wise Bregman representative
\begin{equation}
{D}_{IS}(Q, \p)  \approx \frac{1}{M}\sum_{m = 1}^M \frac{Q}{p_m} - {1 + \mathbb{E}_v[ \gamma]}, \label{eq:P2IS_0}
\end{equation}
where given the measurement vectors $\{\a_m \}_{m = 1}^M$ and a distribution over $\v$, $\mathbb{E}_v[ \gamma] = \mathbb{E}_{\v}[\phi(\p) - \phi(Q)]$. 
In fact, utilizing the properties of Bregman representation we have the expected divergence in the form of:
\begin{equation}
\mathbb{E}_{\v}(d_{\phi} (Q, \s)) = \mathbb{E}_{\v}[\phi(Q)] - \phi(\s) - \mathbb{E}_{\v}[\langle Q - \s , \nabla \phi(\s) \rangle ].
\end{equation}
where by the first order optimality condition, the minimizer over $\s$ satisfies the Jensen gap in \eqref{eq:JensenGap} as
\begin{equation}\label{eq:bregmin}
\text{min}_{\s \in \mathit{S}}\mathbb{E}_{\v}(d_{\phi} (Q, \s)) = \mathbb{E}_{\v}(d_{\phi} (Q, \hat{\q})) = \mathbb{E}_{\v}[\phi(Q)] - \phi(\hat{\q}).
\end{equation}
%Notably, this crucial distinction facilitates the formulation of a spectral method that minimizes Itakura-Saito divergence under minimal Bregman representation. 
Hence via \eqref{eq:bregmin}, the Bregman representative is the best estimator for $\mathbb{E}_{\v}[\phi(Q)]$. 
Then, from the separability of the terms in \eqref{eq:P2IS_0} under the expectation facilitate a formulation of $P2$ %under minimal Bregman representation for minimizing Itakura-Saito divergence.
using the Bregman representative $\hat{\q}$ in \eqref{eq:P2IS_0} as:
\begin{align}
&D_{IS}(\q, \p) \approx \frac{1}{M}\sum_{m = 1}^M q_m \left( \frac{1}{p_m} - \frac{1 + \hat{\gamma}}{\hat{q}_m} \right), \label{eq:P2IS_1} \\
&\hat{\x}_{IS}^{(\hat{\gamma})} := \underset{\| \v \| = 1}{\text{argmax}} \ \v^H \left[\sum_{m = 1}^M \left( \frac{(1 + \hat{\gamma})}{\hat{q}_m} - \frac{1}{p_m} \right)  \a_m \a_m^H \right] \v,\label{eq:ItaSai_gmma}
\end{align}
where $\hat{\gamma} = \phi(\p) - \phi(\hat{\q})$. 
\begin{proposition}
Let $P2$ in \eqref{eq:dmineq2} be constructed as \eqref{eq:ItaSai_gmma}, such that for any $\p \in \mathrm{si}(S)$, $h_{\phi}(\p, \hat{q})$ denotes the sample processing function as:
\begin{equation}\label{eq:optIS}
[h_{\phi}(\p, \hat{q})]_m := [\mathcal{T}(\y)]_m = \frac{(1 + \hat{\gamma})}{\hat{q}_m} - \frac{1}{p_m}
\end{equation}
where $S$ is the domain of choice for $Q$, with $\hat{\q}$ denoting its Bregman representative per Definition \ref{def:BregRep}. 
Then \eqref{eq:ItaSai_gmma} is a spectral method under minimal Bregman representation such that \eqref{eq:P2IS_1} is on expectation the best approximation for the original Itakura-Saito divergence in \eqref{eq:ISdiv} with
\begin{equation}
\big| \mathbb{E}_{\v} [- \langle Q,  h_{\phi}(\p, \hat{\q}) \rangle] - \mathbb{E}_{\v} [D_{IS} (Q, \p)] \big| = \mathbb{E}_{\v}[D_{IS}(Q, \hat{\q})].
\end{equation}
\end{proposition}
\begin{proof}
Considering the approximate form in \eqref{eq:P2IS_1} we have
\begin{align}
- \langle \q,  h_{\phi}(\p, \hat{\q}) \rangle &= \sum_{m = 1}^M \frac{q_m}{M} \left( \frac{1}{p_m} - \frac{1 + \hat{\gamma}}{\hat{q}_m} +  \frac{1 + {\gamma}}{{q}_m} -  \frac{1 + {\gamma}}{{q}_m} \right) \\ 
&=  D_{IS}(\q, \p) +  \frac{1}{M}\sum_{m = 1}^M 1 + \gamma - \frac{q_m}{\hat{q}_m}(1+\hat{\gamma}),
\end{align}
where $\hat{q}_m > 0$ since $\hat{\q} \in \mathrm{ri}(S) \subset \mathbb{R}_{++}^M$. 
Then, taking the expectation over the random variable $Q$, we have
\begin{align}
\mathbb{E}_{\v}[- \langle \q,  h_{\phi}(\p, \hat{\q}) \rangle  ] &= \mathbb{E}_{\v}[D_{IS}(\q, \p)] + 1 + \mathbb{E}_{\v}[\gamma] \\
 &- \frac{1}{M}\sum_{m = 1}^M \mathbb{E}_{\v}[\frac{q_m}{\hat{q}_m}(1 + \hat{\gamma})],
\end{align}
which by definition yields
\begin{equation}
\mathbb{E}_{\v}[- \langle \q,  h_{\phi}(\p, \hat{\q})] = \mathbb{E}_{\v}[D_{IS}(\q, \p)] + (\mathbb{E}_{\v}[\gamma] - \hat{\gamma}).
\end{equation}
Furthermore, from the definition of 
$\gamma$ the distortion corresponds to the Jensen gap, which \eqref{eq:JensenGap} satisfies \eqref{eq:bregmin}. 
Thereby \eqref{eq:P2IS_1} is the best approximation on the original divergence in \eqref{eq:ISdiv} and the principle of minimal Bregman representation is satisfied per Definition \ref{def:MinRep}, with \eqref{eq:optIS} as the optimal sample processing function for minimizing the Itakura-Saito divergence per Definition \ref{def:OptSP}. 
\end{proof}

We remark that there is a key insight of the framework present in our formulation of \eqref{eq:IS_spec}. 
Observe that in obtaining $P2$ under minimal Bregman representation, the $\p$ dependent term $\phi(\p)$ in $\gamma$ could be discarded into $c_{\phi}(\p)$ along with the ``DC'' term of $-1$. 
However, such a discard has no motivation towards attaining the minimal Bregman representation, as with or without the inclusion of these terms in the spectral method the order of approximation using the Bregman representative is identical. 
On the other hand, via the full inclusion of $\gamma$, the original form of the Bregman divergence is preserved using the minimal distortion estimate for the scaling between the reciprocal terms within $h_{\phi}$. 
Thereby, a minimal Bregman representation is obtained with $c_{\phi}(\p) = 0$.

Itakura-Saito is a divergence over any subset in $\mathbb{R}_{++}^M$, i.e., any $\q, \p$ where $q_m > 0, p_m > 0, \forall m = 1, \cdots M$.
According to the specification of the domain of $Q$ in $\mathbb{R}_{++}^M$, \eqref{eq:ItaSai_gmma} yields different spectral methods. 
Setting the domain of $Q$ as the image of the set of rank-1, PSD matrices in the range of $\mathcal{A}$, we obtain the ideal Itakura-Saito divergence minimizing spectral method with $P2$ using the Bregman representative $\hat{\q}^{(2)}$ in \eqref{eq:rep2} with $\p = \y/\| \x \|^2$, where
\begin{equation}\label{eq:IS_spec}
h_{\phi}(\p, \hat{\q}^{(2)}) = \frac{(1 + \phi(\y/\| \x \|^2) - \phi(\hat{\q}^{(2)}))}{\| \a_m \|^2 / N} - \frac{\| \x \|^2}{y_m}.
\end{equation}
Clearly for \eqref{eq:IS_spec} to be realizable, we need the ideal on normalization on $\y$ so that $\p \in S$ for the original divergence. 
Alternatively, considering the $M-\mathrm{Simplex}$ as the domain of $Q$ for generic applications when a reliable estimate for $\| \x \|^2$ is unavailable, appropriately setting $\p = \y/ \| \y \|_1$ and $\hat{\q} = \hat{\q}^{(3)}$ yields a sample processing function of:
\begin{equation}
h_{\phi}(\p, \hat{\q}^{(3)}) = \frac{(1 + \phi(\y/\| \y \|_1) - \phi(\hat{\q}^{(3)}))}{\| \a_m \|^2 / (\sum_{i = 1}^M \| \a_i \|^2 )} - \frac{\sum_{i = 1}^M y_i}{y_m}. \label{eq:SP_Msimplex}
\end{equation}
Consequently, on the $M-\mathrm{Simplex}$ the Itakura-Saito divergence is minimized upto a scale of $\| \q \|_1$ as encountered in the KL-divergence minimizing spectral method, since the domain of the minimization in $P2$ undergoes a mapping by as $\| \q \|_1 = \| \mathbf{A} \v \|_2^2$.

Notably, for the Gaussian sampling model both spectral methods are feasible, and the two formulations reach the identical form with high probability as:
\begin{equation}
\frac{1}{M}\sum_{m = 1}^M \| \a_m \|^2 \approx N, \quad \frac{1}{M} \sum_{m = 1}^M y_m \approx \| \x \|^2. 
\end{equation}
Furthermore, observe that in the Gaussian sampling setting with i.i.d. samples from $\a_m \sim \mathcal{CN}(0,1)$, $\gamma$ takes the form of a \emph{sample mean}, where for large $M$ we have
\begin{equation}
\gamma \rightarrow \mathbb{E}_{\a_m \sim \mathcal{CN}(0,1)}\left[ \log \frac{| \langle \a_m, \v \rangle |^2}{ | \langle \a_m , \x/\| \x \| \rangle |^2}\right],
\end{equation}
for any $\v$ with $\| \v \| = 1$ from the strong law of large numbers. 
Then, from the unitary invariance of Gaussian distribution, we obtain an estimate of $\hat{\gamma} = 0$ without invoking the Bregman representative. 
Subsequently the spectral method reduces to:
\begin{equation}
\hat{\x}_{IS}^{(0)} := \underset{\| \v \| = 1}{\text{argmax}} \ \v^H \left[\sum_{m = 1}^M \left( \frac{1}{\hat{q}_m} -\frac{1}{p_m} \right)  \a_m \a_m^H \right] \v.\label{eq:ItaSai_gmma}
\end{equation}

\subsection{Itakura-Saito Optimality in the Gaussian Sampling Model}

%In this paper, we repeatedly consider $M-\mathrm{Simplex}$ as the domain over which the random variable $Q$ in formulating our spectral method of $P2$. 
%Beyond the motivation of absolute normalization, the $M-\mathrm{Simplex}$ provides an advantage even in the Gaussian sampling model in which an estimate for the energy of the signal is present with high probability. 
%Namely, the scaling that perturbs the ideal spectral methods on the $M-\mathrm{Simplex}$ is governed by an $\ell_1$ isometry on the lifted forward model $\mathcal{A}$. 
%Whereas for the case of  $\ell_2$ loss minimizing formulation, the perturbation from the scaling is governed by an $\ell_2$ isometry which is known require a heavier tail bound, hence a higher sample complexity. 
%To this end, it should be expected that the KL and IS-divergence minimizing spec

\eqref{eq:ItaSai_gmma} reveals quite an interesting outcome of our formulation under the Gaussian sampling model.
Setting $\p$ as the normalized measurements, i.e., $\p = \y/ \lambda_0$, from the strong law of large numbers we have that $\gamma \rightarrow \mathbb{E}(\gamma)$ as $M \rightarrow \infty$, where the expectation is computed over the sampling vectors. 
Thereby, given any unit vector $\v$ and having i.i.d. standard complex Gaussian distributed $\a_m$ for $m = 1, \cdots M$ where $\alpha = M/N$ is fixed, $\gamma$ approaches to $0$ in the limit that $N \rightarrow \infty$ since $\mathbb{E}[\phi(\p)] = \mathbb{E}[\phi(\q)]$ through the invariance of spherical distributions to orthogonal transformations. 
As a result, the initial approximation of the Itakura-Saito divergence in \eqref{eq:P2IS_0} is asymptotically exact and tight with high probability in finite dimensions. %, whereas the following approximations become redundant. 
Then, using the Bregman representative $\hat{\q}^{(1)}$ as derived under the statistical assumptions of Gaussian sampling vectors, yields the following spectral method:
\begin{equation}
\hat{\x}_{IS}^{(0)} := \underset{\| \v \| = 1}{\text{argmax}} \ \v^H \left[\sum_{m = 1}^M \left(1 - \frac{1}{p_m} \right)  \a_m \a_m^H \right] \v.\label{eq:ItaSai_v0}
\end{equation}

Namely, the Itakura-Saito divergence minimizing spectral method precisely corresponds to the limiting optimal sample processing function that is derived asymptotically $N \rightarrow \infty$ in \cite{luo2019optimal}. 
%In addition, observe that the optimal sample processing function derived in \cite{mondelli2017fundamental} has a similar premise, except that $\p$ is a smoothed version of $\y/\lambda_0$ towards its expected value of $\mathbb{E}[\y/\lambda_0] = 1$ as
In reverse, the result of \cite{luo2019optimal} confirms that the framework we formulate by using the Bregman representative is the optimal manner to approximate the Itakura-Saito divergence for spectral estimation under the Gaussian sampling model assumption.
Indeed, observe that
\begin{align}
-\langle \q, h_{\phi}(\p, \hat{\q}) \rangle &= \frac{1}{M} \sum_{m = 1}^M \left( \frac{1}{\hat{q}_m} - \frac{1}{p_m} + \frac{1}{q_m} - \frac{1}{q_m} \right) q_m \\
&= \frac{1}{M} \sum_{m = 1}^M  \left( 1 - \frac{q_m}{p_m} \right) + \frac{\frac{1}{M}\sum_{m =1}^M q_m}{\mathbb{E}[q_m]} - 1
\end{align}
which is a consistent estimator of the ideal divergence as $\gamma \rightarrow 0$.
%From spectral estimation literature we have the following property that

However, the particular setting of phase retrieval under the Gaussian sampling model highlights a rather interesting property. Under our general formalism, the best tractable spectral method we have for Itakura-Saito divergence minimization corresponds to using a sample processing function that is optimal over \emph{all possible choices} for $\mathcal{T}$ according to the analysis in \cite{luo2019optimal}. 

To this end, we revisit well established results in information theory and spectral density estimation literature that support this conclusion \cite{gray1980distortion}, which we utilize in the following proposition.
\begin{proposition}\label{prop:ARM}
Let $M$ phaseless measurements of the form $y_m = | \langle \a_m , \x \rangle |^2$, $m = 1, \cdots M$ be collected under a Gaussian sampling model, i.e., with $\a_m \sim \mathcal{CN}(0, \mathbf{I}_N)$ with $M = \alpha N$, where $\alpha > 1$ is fixed. 
For a $\v \in \mathbb{C}^N$ with $\| \v \|=1$, let $q_m = | \langle \a_m , \v \rangle |^2$. Then, letting $N \rightarrow \infty$, the spectral method in \eqref{eq:ItaSai_v0} with $p_m = M y_m / (\sum_{m}^M y_m)$ equivalently minimizes the discrimination information $I(P, Q)$ in \eqref{eq:limcri}, where $P$ and $Q$ are the power spectral densities of two zero mean Gaussian processes that have an auto-regressive model that is parameterized through $\x$ and $\v$ as:
\begin{align}
\tilde{a}_P[k] &= \frac{1}{\sqrt{M}} \sum_{m  = 1}^M \mathrm{e}^{\mathrm{j}2\pi k m/M} \a_m^H \x,\  k = 1, \cdots M \\
\tilde{a}_Q[k] &= \frac{1}{\sqrt{M}} \sum_{m  = 1}^M \mathrm{e}^{\mathrm{j}2\pi k/M} \a_m^H \v, \ k = 1, \cdots M
\end{align}
which have the $M$-point periodograms $A_P^M$, and $A_Q^M$, respectively, such that $A_P^M \rightarrow A_P(\mathrm{e}^{\mathrm{j} \omega})$, $A_Q^M \rightarrow A_Q(\mathrm{e}^{\mathrm{j} \omega})$ as $M \rightarrow \infty$ 
\begin{equation}\label{eq:autoreg}
P = \| \x \|^2 / | A_P(\mathrm{e}^{\mathrm{j} \omega}) |^2, \quad Q = 1 / | A_Q(\mathrm{e}^{\mathrm{j} \omega}) |^2.
\end{equation} 
%Thereby, the spectral method in \eqref{eq:ItaSai_v0} is asymptotically optimal as t
\end{proposition}
\begin{proof}
In the proof, we use an abstract equivalence to Gaussian processes, which is based on the almost sure existence of a valid auto-regressive model of the form of \eqref{eq:autoreg}. Namely, given phaseless measurements $\y$ realized by Gaussian sampling, $\p$ and $\q$ are analogous to gain normalized periodograms via the unitary invariance property $\{ \a_m \}_{m = 1}^M$. 
Using the orthonormal $M$-point DFT matrix $\mathbb{W}_M$, we have that
\begin{equation}\label{eq:PSD}
\y := | \mathbf{A} \x |^2 = | \mathbf{W}_M \mathbf{W}^H_M \mathbf{A} \x |^2 = | \mathbf{W}_M \tilde{\mathbf{A}} \x |^2
\end{equation}
where $\tilde{\A}$ is a Gaussian sampling model via unitary invariance, and $\alpha > 1$ ensures the underlying linear map on the parameter domain is injective such that the search space of the spectral method facilitates a search of the periodogram generating samples in a manner that is one to one. 

%Furthermore, any rank-1, PSD element maps to an element in the range of $\mathcal{A}$ that is strictly positive almost surely. 
Furthermore, let $g[m]$ denote the underlying $M$ samples that are generated by $\tilde{\A}$, which are zero-mean Gaussian i.i.d. by definition. 
Then, the reciprocal spectra is also generated by zero-mean i.i.d. Gaussian samples since $\mathbb{E}[g[m] \ast g^*[-m]] = \delta[m]$, which the sample auto-correlation tends to as $M \rightarrow \infty$. 
Thereby,
\eqref{eq:PSD} describes auto-regressive models for two zero mean Gaussian processes. 
%As the original periodograms are generated via i.i.d. Gaussian samples, the inverse spectra relates the original samples via a linear filter which preserves Gaussianity. 
Using the reciprocal symmetry of the Itakura-Saito divergence yields
\begin{equation}\label{eq:limIS}
D_{IS}(P, Q) = D_{IS} (| A_Q(\mathrm{e}^{\mathrm{j} \omega}) |^2, \frac{|A_P(\mathrm{e}^{\mathrm{j} \omega}) |^2}{\| \x \|^2}) = \lim_{M \rightarrow \infty} D_{IS}( \q, \p ).
\end{equation}

Then, for two zero mean Gaussian processes with spectral densities of $P$ and $Q$, we know \eqref{eq:limIS} tends to their discrimination information from \eqref{eq:limcri}, where the limit is  
is identical to that of the forward KL-divergence between the probability distributions of $M$-point ensembles as \eqref{eq:infcri}, with
\begin{equation}
\lim_{M \rightarrow \infty} D_{IS}( \q, \p ) = \lim_{M \rightarrow \infty} \frac{2}{M} I_M(P, Q) = 2 I(P, Q).
\end{equation}
Hence the proof is complete. 
\end{proof}

Accordingly, Proposition \ref{prop:ARM} connects the asymptotic result of \cite{luo2019optimal} to the more general setting of our framework, where it is an instance of a metric optimal spectral method. 
The global optimality is then a result of the asymptotic characterization of the Itakura-Saito divergence as a distortion measure that captures the discrimination information of two zero mean Gaussian processes. 
On the other hand, knowing that \eqref{eq:ItaSai_gmma} is the best approximation of Itakura-Saito divergence on expectation over $\v$, given a set of sampling vectors $\{ \a_m \}_{m = 1}^M$ that were realized under the Gaussian sampling assumption our framework implies a finite sample correction on the optimally derived sample processing function of \cite{luo2019optimal}.

\subsection{Relations to Orthogonality Promoting Methods}

As studied in Section \ref{sec:OrthMinnrom}, leveraging orthogonality information in the phaseless measurements is a fundamental approach investigated by several methods in the literature. 
Notably, the primary motivation of the framework we introduce in this paper is to analyze the internal mechanisms of spectral methods.
We achieve this through a tractable loss minimization perspective by utilizing universal properties of Bregman divergences. 
However, the robust loss based construction of orthogonality promoting methods is not governed by our abstract framework, as the $\ell_1$ norm is not a strictly convex function for the purpose of generating a Bregman divergence. 

Nonetheless, the spectral methods we optimally derive in this section have a connection to orthogonality promoting methods. 
This is direct consequence of the functional forms we obtain for the sample processing functions, through which, both the KL and IS minimizing spectral methods feature an \emph{indicator} function of the orthogonal plane to $\a_m$ as $y_m \rightarrow 0$. %, both the KL and IS divergence minimi
Thereby, in the presence of a measurement $y_m$ that encodes strict orthogonality, both methods penalize a non-zero synthesis of $\q_m$ with $-\infty$, which in turn constrains an estimate to align orthogonal to the corresponding sampling vector. 
Interestingly, this is a universal property that is beyond the scope of the original orthogonality promoting methods in the literature.
To similar effect, the aggressive penalization of these methods is reflected in the presence of near-orthogonality indicating measurements.
Hence, both spectral methods effectively promote orthogonality with respect to subsets of sampling vectors which are determined according to a \emph{score} assigned by the sample processing function.

Towards investigating the relation to orthogonality promoting methods in the literature, consider the \emph{reward} of an ideal match in the range of $\mathcal{A}$. 
Under strong law of large numbers, with the Gaussian sampling model the sums $\sum_{m} \a_m \a_m^H$ are tightly concentrated \cite{duchi2019solving}, and the original orthogonality promoting formulation of \cite{wang2018phase} has the following equivalence:
\begin{align}\label{eq:ell_1_0}
\sum_{m \in \mathit{I}_{S<}} \frac{1}{\| \a_m \|^2} \a_m \a_m^H &+ \sum_{m \in \mathit{I}^c_{S<}} \frac{1}{\| \a_m \|^2} \a_m \a_m^H \\
&= \sum_{m = 1}^M \frac{1}{\| \a_m \|^2} \a_m \a_m^H \approx \frac{M}{N} \mathbf{I}.
\end{align}
Essentially, the smallest eigenvalue-eigenvector pair that the orthogonality promoting method seeks from the first component in the left-hand side of \eqref{eq:ell_1_0}, is approximately attained by the leading eigenvector of the second term.
Thereby a spectral matrix can be synthesized using the complementary index set $\mathit{I}^c_{S<}$ of sampling vectors that have \emph{maximal correlation} to the ground truth. 
This is further generalized to powers of the intensity measurements as $y^k_{m \in \mathit{I}^c_{S<}}$ where $0 \leq k \leq 1$ which are introduced as weights to the synthesis of the spectral matrix \cite{wang2018phase}, with the orthogonality promoting method as a special case $k = 0$. 

Under the methodology of \cite{wang2018phase}, we can simply consider the reward of matching a measurement $y_m$ for several different methods including a non-orthogonality promoting ones such as the classical spectral method as
\begin{equation}
R_{0}(y_m) = y^{k+1}_m, \  0 < k < 1. 
\end{equation}
Clearly, sample truncation, or pruning are implemented to control this reward as $y_m$ gets large via clipping, or analogously damping the reward towards zero. 
On the other hand the ideal orthogonality promoting method pursues the minimal reward feasible via $k \rightarrow 0$, however in doing so loses all information in the measurements when synthesizing $\Y$ as $y_m^0 = 1$.
Accordingly, the lost information is installed into the synthesis via subset selection.  
In this sense, we concur that the orthogonality promoting method imitates the $\ell_1$ loss minimization through its \emph{linear} reward for a match in the phaseless measurements. 

Observe that this is precisely achieved by the Itakura-Saito divergence in \eqref{eq:ItaSai_v0}, where the reward of a match is as:
\begin{equation}
R_{IS}(y_m) = y_m - 1.
\end{equation}
Furthermore, an $\mathcal{O}(y_m)$ reward is assigned to all measurements at synthesis without any assumptions on the measurement model, or hand-tuning for subset selection as no information is lost at the synthesis.
Hence, an estimate is formed under the desired linear reward system for promoting orthogonality. 
%On the other hand, given that the proper normalization defined within our formalism is applied such that $\mathbb{E}[y_m] = 1$ under Gaussianity, the IS-divergence reward is the tangent of the KL-divergence minimizing method that assigns the more orthogonality promoting reward.
%As demonstrated in Figure \ref{} the KL-minimizing method rewards matching $y_m = 0$ better than matching measurements where $y_m \leq \mathbb{E}[y_m]$. 
%Although the reward grows as $y_m \log y_m$, matching a measurement with $\mathbb{E}[y_m] \geq y_m$ is noticeably pruned compared to that of the IS divergence until around $y_m \geq 6\mathbb{E}[y_m]$. 
%Yet another interesting observation is that the $\alpha$-adaptive scheme of \cite{mondelli2017fundamental} follows a similar reward to our KL-minimizing method, especially in the $y_m \leq \mathbb{E}[y_m]$ regime near information theoretic limits of $M \approx 4N$, with stronger penalization when $M \leq 2N$. 
%In a similar manner, the smoothing effect outlined in \eqref{eq:preproc} can be interpreted as rewarding orthogonality more in favor of large sample values, especially when the estimation is severely sample starved near $M \approx N$ with $M > N$. 

%in a manner that it almost \emph{interpolates} between the reward of IS minimizing and KL minimizing methods.  

Finally, we note that the KL-divergence minimizing formulation shares a fundamental connection to the orthogonality promoting method of \cite{wang2018solving} in that both pursue a dissimilarity to the flat signal, or analogous to our formulation, the uniform distribution. 
This relation stems from the fact that the methodology of \cite{wang2018solving} can be interpreted as searching for an estimate that synthesizes the measurements that are least aligned to the uniform distribution in a Euclidean-sense. 
To this end, our formulation utilizes the more natural metric of dissimilarity via penalizing KL-sense similarity to uniform distribution, i.e., the Bregman representative within its objective. 
%Essentially, it can interpreted that extracting the eigenvector associated with the minimum eigenvalue of the first term in \eqref{eq:ell_1_0} as an operation under the classical $\ell_2$ alignment measure, i.e., the method synthesizes measurements in the range of $\mathcal{A}_{I_S}$ that are the most dissimilar to the flat signal, i.e., $y_m = 1$, or under normalization, \emph{the uniform distribution}.
%To this end, using the Bregman representative under the Gaussian sampling model, the surrogate loss function that is minimized using the KL-optimal sample processing function features a penalty that promotes dissimilarity to the uniform distribution using the more appropriate metric of KL-divergence. 
Most significantly, such behaviors that relate to orthogonality are model independent under our framework, and orthogonality is promoted only to a degree that benefits the overall scheme at the synthesis equation. In essence, the dissimilarity from the uniform distribution is promoted independent of an inherent assumption on the sampling vectors, and captures a more general principle in the estimation task than those encountered in the literature.  

\section{Numerical Simulations}\label{sec:4_sec5}
In this section, we assess the practical effectiveness of our framework via a detailed comparison of the derived Bregman-divergence minimizing schemes with several prominent methods in literature. 
In addition to an overall comparison, we particularly highlight the impact of the modifications suggested on existing methods under our formulation. 
%These modifications include:
%\begin{itemize}
%\item  \emph{The $\ell_2$-loss minimizing spectral method under the Gaussian sampling model.} As stressed in its derivation, we utilize the existence of a statistical estimator for the norm of the underlying signal of interest in \eqref{eq:l2min}. Particularly we evaluate two variations on \eqref{eq:l2min} using the Bregman representatives of $\hat{\q}^{(1)}$ and $\hat{\q}^{(2)}$ in comparison to the classical spectral method, which we identically derive in \eqref{eq:} as a sub-optimal approximation under our framework.    
%
%\item \emph{The Itakura-Saito divergence minimizing spectral method.} 
%\end{itemize}

\subsection{Simulation Setup}
We evaluate the quality of estimation via the correlation coefficient, i.e.,
\begin{equation}\label{eq:corr_coef}
\rho( \x , \hat{\x}) = \frac{| \langle \x , \hat{\x} \rangle |}{\| \x \|}
\end{equation}
where $\hat{\x}$ is the unit-norm estimate of $\x$ produced by the spectral method of evaluation given the measurement vectors $\{ \a_m \}_{m = 1}^M$.
Particularly, we use Monte-Carlo simulations for our assessment of each choice of $\mathcal{T}$ in the spectral method, where for a fixed unknown $\x$ we independently pick the $M = \alpha N$ number of measurement vectors to evaluate an empirical average of \eqref{eq:corr_coef} over $l = 1, \cdots L$ instances of the measurement map. 
Given $s = 1, \cdots S$ different ground truths, the overall performance of each spectral is expressed as:
\begin{equation}\label{eq:empr_corr}
\rho_e(\mathcal{T}) := \frac{1}{SL}\sum_{s = 1}^S \sum_{l = 1}^L \rho(\x^{s}, \hat{\x}_{\mathcal{T}}(I^{s,l}) )
\end{equation}
where $I^{s,l} = (\x^s, \{ \a_m^l \}_{m = 1}^M)$ denote the $(s, l)^{th}$ realization of the input object fed to the spectral method in generating the estimate for the sample $\x^s$ using the pre-processing function $\mathcal{T}$.  

In our numerical simulations, we use the implementations provided by the PhasePack library in \cite{chandra2019phasepack} for the methods from prior literature that are available. 
These include the implementations of the classical and truncated spectral methods, the orthogonality promoting method, and the weighted maximal correlation method, as well as the first optimally derived sample processing function of \cite{mondelli2017fundamental} under the proper normalization of measurements. 
%The details of their implementation are summarized by Algorithm \ref{}. 
For the purpose of our comparisons, we integrate our modifications on PhasePack to implement the optimal sample processing function \cite{luo2019optimal}, and the Bregman divergence minimizing methods we derive in Section \ref{sec:4_sec3}.

It should be noted that we include our own implementation for spectral estimation via the minimum norm solution as outlined in Section \ref{sec:4_sec1}, and not the closely related linear spectral estimators designed and evaluated in the original work of \cite{ghods2018linear} in relevance to our framework and its comparisons.
We provide Table \ref{table:specmethods} which associates each spectral method with its corresponding pre-processing function, $\bar{y}_m$ denoting the mean-normalized measurements, $i(\cdot)$ denoting the indicator function of its input indices, $\mathit{I}_{S = 5M/6}$ denoting the index set of lowest $5M/6$ magnitude measurements, and $\mathit{I}^{c}$ denoting its complement, where $5M/6$ is picked under the specifications of the source work in \cite{wang2018solving}.  

\begin{table}[!htbp]
\caption{Table of Spectral Methods} % title of Table
\centering % used for centering table
{
\small
  \begin{tabular}{p{1.3in}p{1.8in}}%p{0.1in}p{0.7in}p{1.4in}}
 
\hline \hline
  Method & $[\mathcal{T}(\y)]_m$ \\ % & &Symbol & Designation\\ [0.5ex] % inserts table
  %heading
  \hline % inserts single horizontal line
%  \\
%\begin{tabular}{c l } % centered columns (4 columns)
%\hline\hline %inserts double horizontal lines
%Symbol& \hspace{0.5in}Description \\ [0.5ex] % inserts table
%heading
%\hline % inserts single horizontal line\
Classical & $y_m$ \\
Truncated & $y_m \cdot i(\bar{y}_m < 3)$ \\
Orthogonal & $1 \cdot i(\mathit{I}_{S = 5M/6})$ \\
Weighted & $ {y^{1/4}_m} \cdot i(\mathit{I}^c_{S = 5M/6})$ \\
Minimum Norm & $[(|\A \A^H |^2)^{-1} \y]_m$ \\
MM Optimal & $ (\bar{y}_m - 1)/(\bar{y}_m + \sqrt{\alpha} - 1)$ \\
LL Optimal & $ 1 - 1/\bar{y}_m$ \\
KL Minimizing & $\log(\bar{y}_m / (M \hat{q}_m))$ \\
IS Minimizing & $ \left( (1 - \hat{\gamma})/\hat{q}_m - M / \bar{y}_m \right) \cdot i(\bar{y}_m > 0)$ \\
$\ell_2$ Minimizing & $2 \bar{y}_m - \hat{q}_m$ \\
\hline \hline
%\br %inserts single line
  \end{tabular}}
  %\end{indented}
  \label{table:specmethods}
\end{table}

\subsection{Gaussian Sampling Data}

We first consider the general statistical setting under the Gaussian sampling model with $N = 256$, where we independently sample i.i.d. entries of $\a_m$ with $A \sim \mathcal{N}(0, 1/2) + \mathrm{j}\mathcal{N}(0, 1/2)$ for $m = 1, 2, \cdots M$. 
This is comparable to the problem size considered for the Gaussian sampling model in \cite{ghods2018linear} for the particular setting of communication systems. 
As such, the $\a_m \in \mathbb{C}^N$ are i.i.d. samples from the $N$-dimensional standard complex Gaussian model. 

%We then set up a Monte-Carlo simulation, where \eqref{eq:corr_coef} is averaged over $50$ samples 
In our assessment via Monte-Carlo simulations, we generate independent instances of the collection of $\{ \a_m \}_{m = 1}^M$ by sampling from $A$ in an inner loop for a fixed underlying signal of interest $\x$. 
The simulation setup features an outer loop that generates a ground truth signal of interest from smoothing a real-valued Gaussian random vector sampled from $X \sim \mathcal{N}(0, \mathbf{I})$ via keeping the central $2B$-DFT coefficients as 
\begin{equation}\label{eq:gtruth}
x_n = \sum_{ k = N/2 - B + 1}^{N/2 + B} \tilde{X}(k) \ \mathrm{e}^{ 2\pi \mathrm{j} n k / N}
\end{equation}
where $\tilde{X}(k)$ denotes the $k^{th}$ shifted DFT coefficient of the underlying sample $\tilde{\x}$ from $X$. 
%An example signal is presented 
For \eqref{eq:gtruth}, the phaseless measurements \eqref{eq:phaless} are synthesized using the collection of $\{ \a_m \}_{m = 1}^M$ sampled from $A$ in the inner loop, and spectral estimation is performed for the sampled object $I := (\x, \{ \a_m \}_{m = 1}^M)$. 

We then proceed with a modestly sized Monte-Carlo simulation with $S = 25$ and $L = 50$, while performing the experiment parametric over the choice of $B$, with $B$ taking the values $[10, 20, 30, 40, 50]$. 
Overall, our numerical simulations indicate reliable performance across the varying bandwidths for each method, with no variation in empirical performance with respect to $B$.  

\begin{figure}
\centering
\includegraphics[scale=0.3]{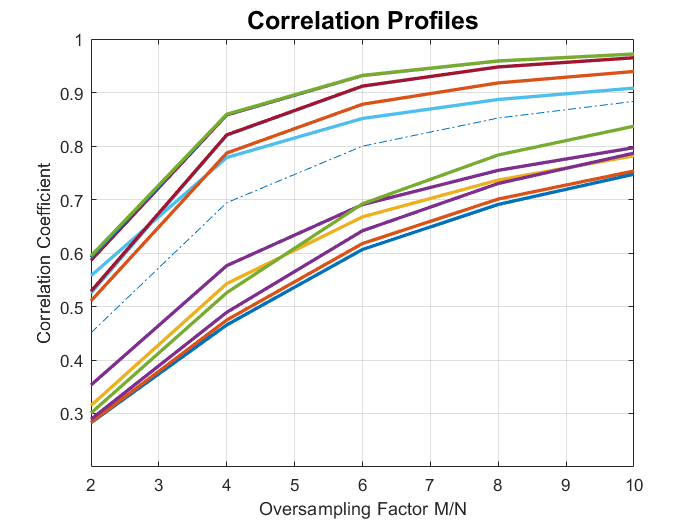}
\includegraphics[scale=0.3]{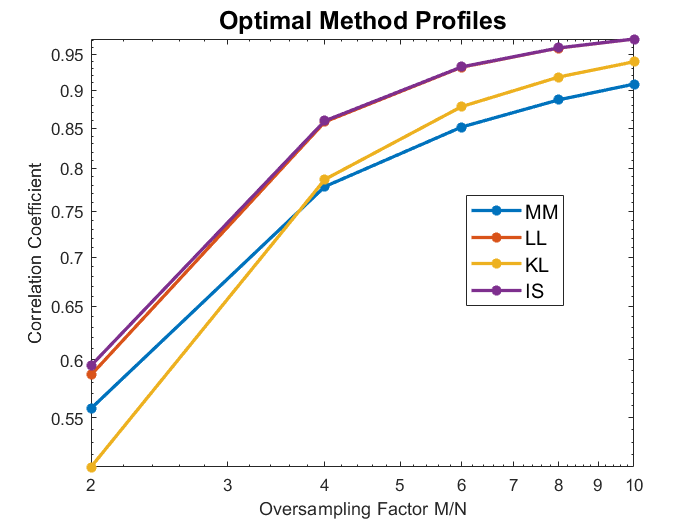}
\caption{\emph{Correlation profiles of each tested spectral method.} (Left) Demonstration of the general trend observed in the spectral methods. Namely, the high correlations are obtained via the optimal methods of $MM$ and $LL$, and our Bregman loss minimizing framework with information theoretic metrics of $KL$ and $IS$ divergences. The methods under the clear divide depicted by the dotted line are $\ell_2$ loss based, or orthogonality promoting methods. (Right) Demonstrates the trend of the optimal spectral methods in logarithmic scale.}
\label{fig:Chp4_fig1}
\end{figure}

%An example signal generated by \eqref{eq:gtruth} for each value of $B$ is demonstrated in Figure \ref{fig:}, and the table of the empirical correlation averages is provided with Table \ref{}. 

Our plots for the correlation profiles averaged over parameter $B$ along with the Monte-Carlo simulation parameters are provided in Figure \ref{fig:Chp4_fig1}. 
The numerical simulations indicate a clear distinction of spectral methods into two classes due to the clustering of the correlation profiles over varying oversampling ratios. 
We provide Table to clearly demonstrate how each method performs with respect to others. 
Ultimately, the optimally designed methods of $MM$, and $LL$ join our novel Bregman divergence minimizing spectral methods in the top cluster, whereas the remaining methods exhibit considerably lower correlation profiles across all tested oversampling factors. 
It should be noted that the minimum norm solution attains the closest performance to the top cluster as oversampling ratios get higher, which is in agreement with the discussion in \ref{sec:OrthMinnrom} and the observations in \cite{ghods2018linear}. 
However, this improvement is not at a level that is sufficient to compete with the optimally designed spectral methods, which practically attain near exact-recovery level by $\alpha = 10$, and yield almost double the correlation of other methods at severely sample starved regime of $\alpha = 2$. 

In the comparison among the optimally designed spectral methods, observe that our $KL$-divergence minimizing formulation performs at a comparable level despite the fact that its counterparts are designed optimally solely for the Gaussian sampling model. 
This highlights the generality of our approach, as our formulation is based on a sense of metric-optimality, without utilizing Gaussianity of the sampling vectors in the design. 
In addition, we use our $IS$-divergence minimizing formulation with the Bregman representative of $\hat{\q}^{(2)}$, which we promote as a finite sample correction for the asymptotically optimal formulation of $LL$. 
Indeed, observe that the $IS$ minimizing spectral method mirrors the performance of the $LL$ method with a \emph{slight improvement}, which supports our argument for the finite sample correction in the Gaussian sampling model. 
Furthermore, pairing with our asymptotic optimality result, \ref{fig:Chp4_fig1} confirms the global optimality of the Itakura-Saito minimizing spectral method for the Gaussian sampling model. 

\subsection{Real Optical Imaging Data}

In order to assess the capability of our framework on real-world problems that are governed by non-Gaussian models, we consider the optimal imaging application and corresponding real data set of \cite{metzler2017coherent}.
In particular, a spatial light modulated (SLM) imaging setup is considered, where an intensity-only $256 \times 256$ detector senses the input images through the propagation medium.
Authors then send amplitude or alternatively phase modulated signals as training data to sense the linear mapping from the image to the detector as a first stage of a double phase retrieval using their approximate message passing-based phase retrieval algorithm. 
The sensed measurement matrices are then utilized to reconstruct images generated by the identical SLM setup which are sensed without phase at the detectors. 

For our simulations, we consider both the $16 \times 16$ and $40 \times 40$ inversion tasks, in which we test the performance of our spectral estimation framework under $\alpha = 5$, and $\alpha = 10$ subsets of the $256 \times 256$ measurements which are chosen at random in each iteration of a Monte-Carlo simulation.
We then provide the average correlation profiles in comparison to other methods in literature provided in Table \ref{table:specmethods}. 
A notable distinction of the two problems beyond their dimensions is the residual error of the provided sensing matrices, where the measurement vectors that map the $16 \times 16$ image have an average $0.02$ normalized residual, compared to an average $0.2$ residual in mapping of the $40 \times 40$ images.
As a result, the $40 \times 40$ image reconstruction problem is operated under considerably more noise, which serves as a good benchmark for the performance of the spectral methods in non-idealized conditions for inference. 

Considering that additive noise is present in the measurements, we also consider the noise-model optimal sample processing functions provided in \cite{luo2019optimal} for Poisson, and Gaussian noise as:
\begin{align}
[\mathcal{T}_{P}(\y)]_m &= \frac{\bar{y}_m - \kappa_0}{\bar{y}_m + 1}, \\ 
[\mathcal{T}_{G}(\y)]_m  &= 1 - 1/\left(\bar{y}_m - \sigma_0^2 + \frac{\sigma_0 \Phi'(\bar{y}_m / \sigma_0 - \sigma_0)}{\Phi( \bar{y}_m / \sigma_0 - \sigma_0)} \right),
\end{align}
where $\kappa_0$ is the scaling factor in the Poisson model, and $\sigma_0$ is the variance of the additive zero-mean Gaussian noise. 
However, it should be stressed that the noise-optimal results of \cite{luo2019optimal} are also under the Gaussian model assumption.
We provide the histogram of the sensed measurement maps in Figure \ref{fig:Chp4_Fig2}, which clearly do not follow a Gaussian distribution, and more so resemble a Laplace distribution. 
For the noise optimal methods, we estimate parameters $\kappa_0$, and $\sigma^2_0$ using the $ML$-estimators on training data to yield $\kappa_0 = 0.989, \sigma^2_0 = 0.008$, and $\kappa_0 = 0.896, \sigma^2_0 = 0.089$ for the $16 \times 16$ and $40 \times 40$ sized problems, respectively. 

\begin{figure}
\centering
\includegraphics[scale=0.3]{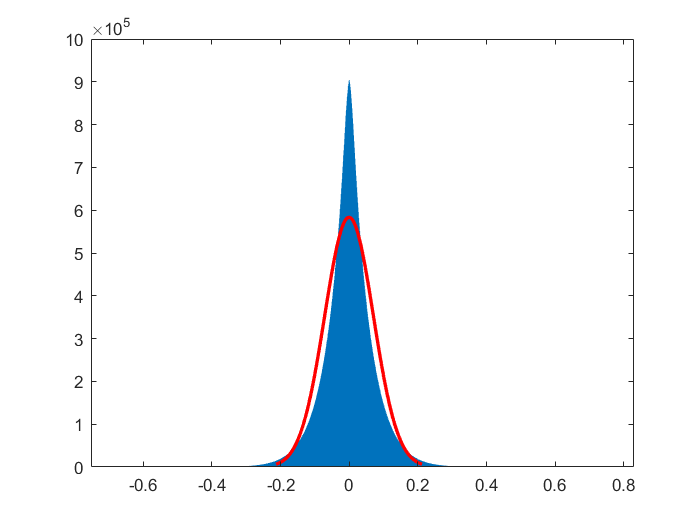}
\includegraphics[scale=0.3]{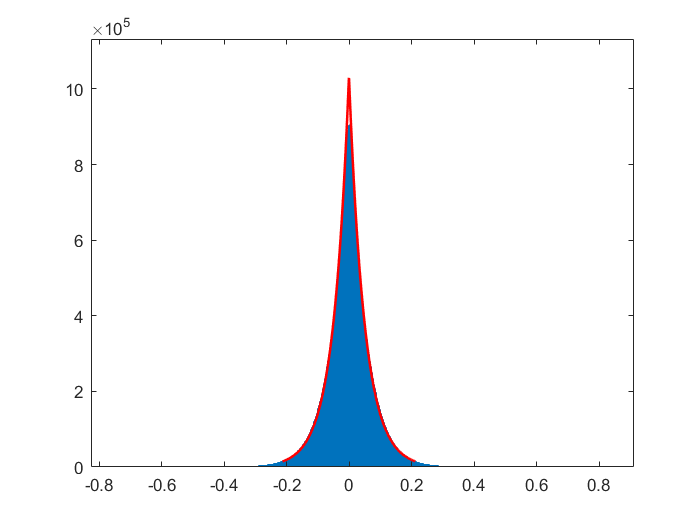}
\caption{Histogram of the provided sensing matrix by \cite{metzler2017coherent} for the $40 \times 40$ SLM, and the Gaussian (left) and Laplace (right) distribution fits to its real-part, using $1600$ bins. Imaginary part mirrors the real part in terms of the fit.}
\label{fig:Chp4_Fig2}
\end{figure}

Finally, we include two versions of the Itakura-Saito minimizing spectral method in order to highlight the significance of our inclusion of the full $\hat{\gamma}$ term in formulating the minimal distortion surrogate. 
%Since in the Gaussian sampling case the optimal formulation featured the cancellation of this term 
Namely, we consider the Itakura-Saito formulation where the influence of $\gamma$ is discarded from the synthesis equation, which is derived as optimal under the Gaussian sampling model assumption,
%we derived under the Gaussian sampling model assumption, where the $\hat{\gamma}$ is optimally set as $0$, 
%a formulation that features the Bregman representative element-wise rather than within the $\hat{\gamma}$ estimate, 
and the hierarchical formulation we propose using the point estimate $\hat{\gamma}$ in \eqref{eq:ItaSai_gmma}.
These correspond to sample processing functions of
\begin{align}
[\mathcal{T}_{IS}^0(\y)]_m = \frac{1}{\hat{q}_m} - \frac{M}{\bar{y}_m}, 
%\quad   [\mathcal{T}_{IS}^1(\y)]_m = \frac{1 + \phi(\y/\|\y\|_1) + \log \hat{q}_m}{\hat{q}_m} -  \frac{M}{\bar{y}_m}, \\
[\mathcal{T}_{IS}^{opt}(\y)]_m = \frac{1 + \hat{\gamma}}{\hat{q}_m} -  \frac{M}{\bar{y}_m}
\end{align}
where we use the Bregman representative of the $M$-Simplex for $\hat{\q}$, with $$\hat{\gamma} = \phi(\y/\|\y\|_1) - \phi(\hat{\q} ).$$

\subsubsection{$\mathbf{16 \times 16}$ Results}

\begin{figure}
\centering
\includegraphics[scale=0.3]{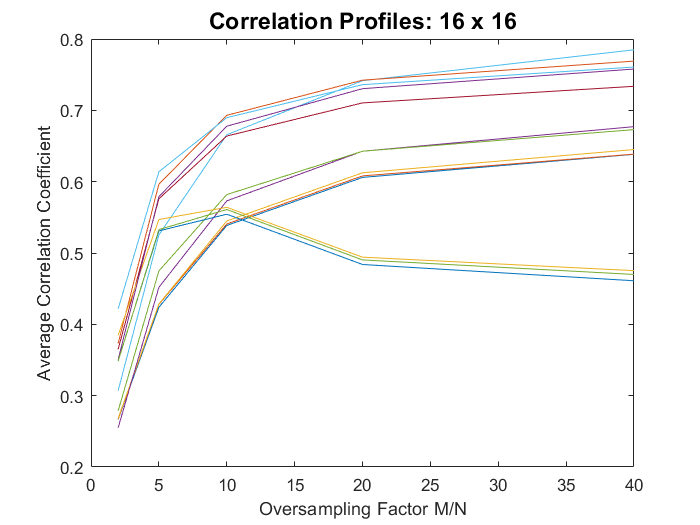}
\includegraphics[scale=0.3]{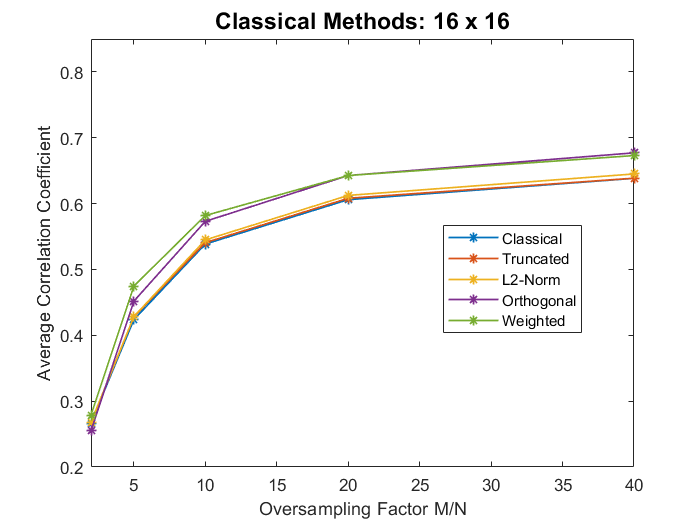}
\includegraphics[scale=0.3]{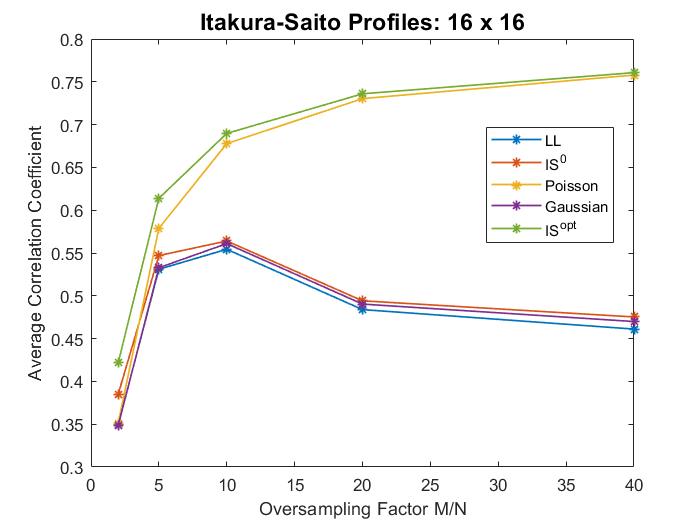}
\includegraphics[scale=0.3]{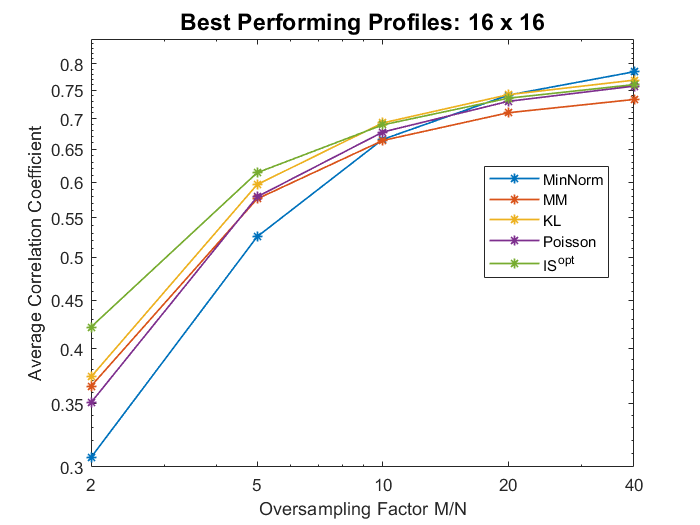}
\caption{Average correlation profiles over the test images, at varying oversampling factors. (Top Left) Full plot including all methods. (Top Right) Classical and orthogonality promoting methods. (Bottom Left) Itakura-Saito related methods. (Bottom Right) Top performing methods in $\log$-scale.}
\label{fig:Chp4_fig3}
\end{figure}

We present our results on the average correlation coefficients computed by \eqref{eq:empr_corr} over the $5$ test images used in our Monte-Carlo simulation at oversampling factors of $[2, 5, 10, 20, 40]$. %in Table \ref{}.
Due to the fact that recovery without phase is only upto a global phase factor, we track the absolute value images and provide the average reconstructions.
We provide the ground truth images in the $16 \times 16$ problem and the average reconstructed absolute value images for the classical spectral method, minimum norm solution, our KL divergence and the optimal IS divergence methods in Figures \ref{fig:Chp4_Im1} to \ref{fig:Chp4_Im5}.
Figures \ref{fig:Chp4_fig3} demonstrate the average correlation profile evolution of all the tested methods with respect to the tested oversampling factors.

\begin{figure*}
\centering
\includegraphics[scale=0.25]{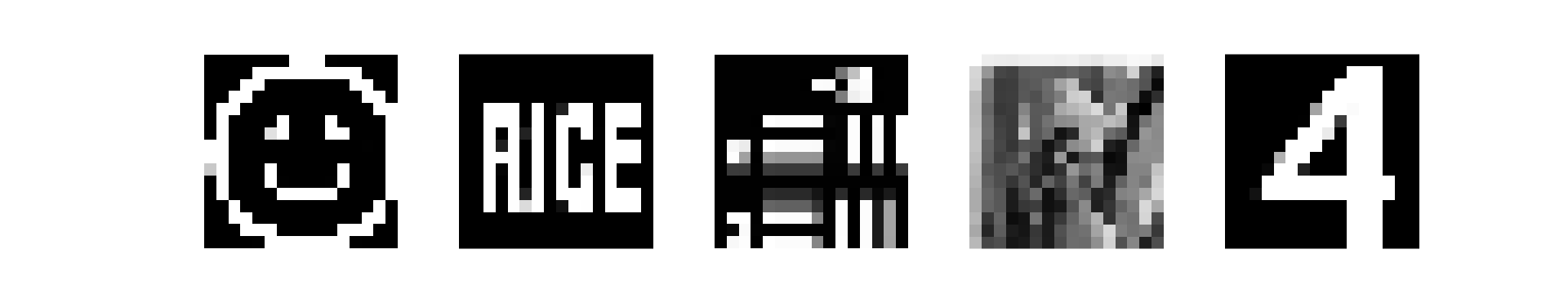}
\caption{Ground truths fed to the SLM imaging setup of \cite{metzler2017coherent}.} 
\label{fig:Chp4_Im1}
\end{figure*}

\begin{figure*}
\centering
\includegraphics[scale=0.25]{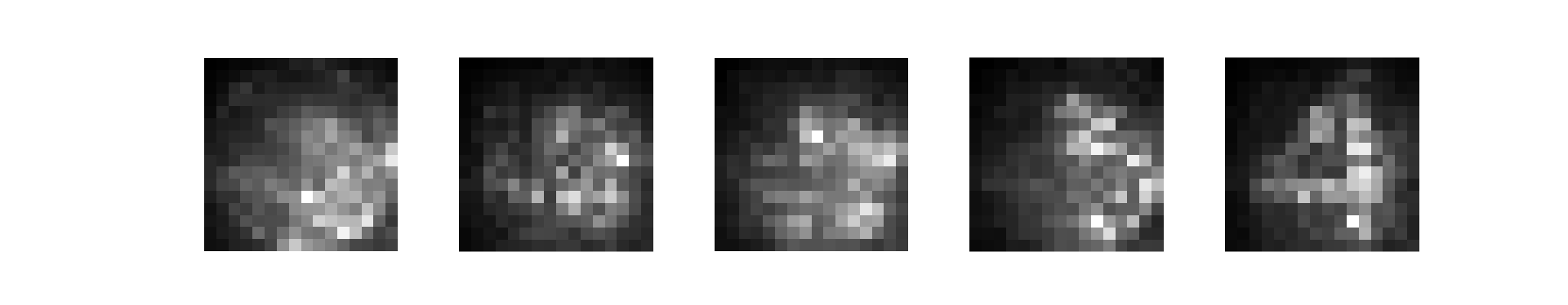}
\caption{Images reconstructed by the classical spectral method. $N = 16 \times 16$, $\alpha = 5$. Average correlations of $ \ [0.3115, 0.4805, 0.2817, 0.4507, 0.6231] \ $ in order.}
\end{figure*}

\begin{figure*}
\centering
\includegraphics[scale=0.25]{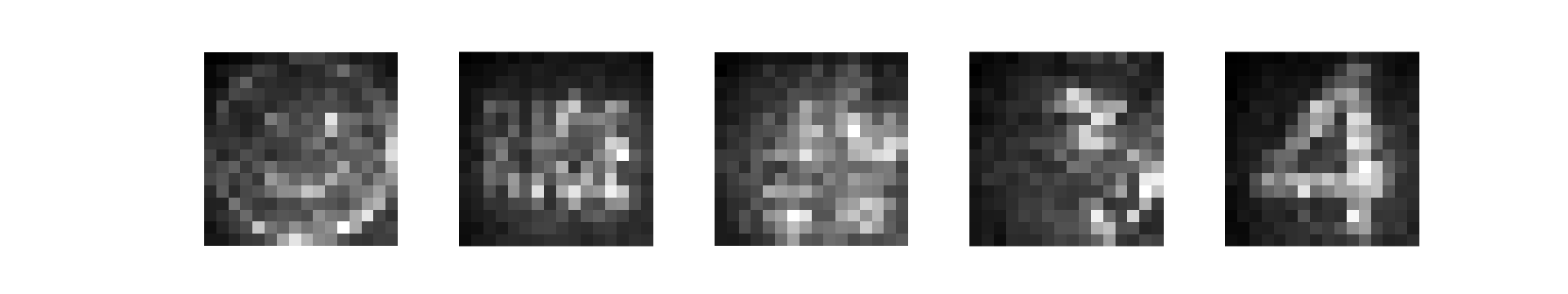}
\caption{Images reconstructed by the minimum norm solution. $N = 16 \times 16$, $\alpha = 5$. Average correlations of $ \ [0.4182, 0.6033, 0.3610, 0.5312, 0.6918] \ $ in order.}
\end{figure*}

\begin{figure*}
\centering
\includegraphics[scale=0.25]{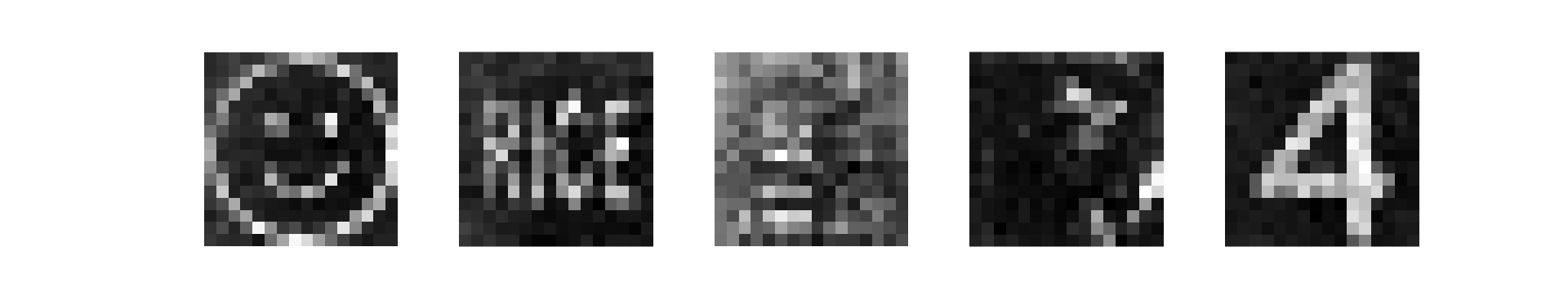}
\caption{Images reconstructed by the KL divergence minimizing spectral method. $N = 16 \times 16$, $\alpha = 5$. Average correlations of $ \ [0.5611, 0.6564, 0.3801, 0.6109, 0.7585] \ $ in order.}
\end{figure*}

\begin{figure*}
\centering
\includegraphics[scale=0.25]{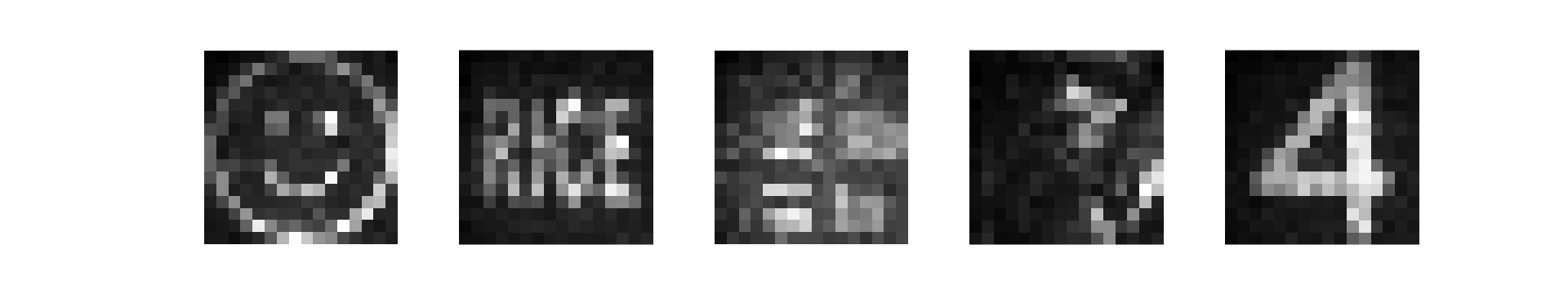}
\caption{Images reconstructed by the IS divergence minimizing spectral method. $N = 16 \times 16$, $\alpha = 5$. Average correlations of $ \ [0.5901, 0.6768, 0.3793, 0.6399, 0.7833] \ $ in order.}
\label{fig:Chp4_Im5}
\end{figure*}

Overall, our results on the smaller, near-noise free $16 \times 16$ sized problem already demonstrate considerable improvement over the state-of-the-art methods in the literature. 
We begin by noting that the classical method, its closely related counterparts of truncation, our $\ell_2$ loss formulation, and orthogonality promoting methods form the middle-of-the pack in terms of average performance. 
It is observed that orthogonality promoting methods provide slight improvement, whereas truncation and our $\ell_2$ loss formulation provide marginal improvement to the classical method. 
On the other hand, it is observed that the Gaussian optimal methods have inconsistent performance as the oversampling factor grows due to the inherent model mismatch in their derivations. 
Notably, $LL$ method, our $IS^{(0)}$ method, and the Gaussian noise optimal methods mirror each other in performance, whereas our Itakura-Saito optimal formulation with the inclusion of $\hat{\gamma}$ denoted as $IS^{opt}$ fully corrects this behavior. 
Interestingly, the Poisson noise model is successful in capturing the model mismatch from the underlying Gaussian assumption, while the Gaussian noise model fails on this end given the low residuals in estimating the sensing model. 
The $MM$ method also provides a consistent performance with its $\alpha$-adaptive, smoothed sample processing scheme, while our $KL$ and $IS$ optimal methods provide the best performance until the minimum norm solution catches up at very high oversampling factors. 
Average absolute value images also demonstrate significant gains in performance with our methods regarding geometric fidelity, and background noise suppression. 

\subsubsection{$\mathbf{40 \times 40}$ Results}

\begin{table}[!htbp]
\caption{Table of Correlations, $40 \times 40$ Images with $\alpha = 5$} % title of Table
\centering % used for centering table
{
\footnotesize
  \begin{tabular}{c | c | c | c | c | c | c }%p{0.1in}p{0.7in}p{1.4in}}
 \hline \hline
  Method & Image $1$ & Image $2$ & Image $3$ & Image $4$ & Image $5$ & Average \\
  \hline % inserts single horizontal line

  Classical & $0.0387$ & $0.0164$ & $0.0277$ & $0.0712$ & $0.0346$ & $\mathbf{0.0377}$	\\
  Truncated & $0.0382$ & $0.0151$ & $0.0272$ & $0.0705$ & $0.0342$ & $\mathbf{0.0370}$ \\
  $\ell_2$-Loss & $0.0388$ & $0.0165$ & $0.0278$ & $0.0712$ & $0.0346$ & $\mathbf{0.0378}$ \\
  Orthogonal & $0.0434$ & $0.0208$ & $0.0307$ & $0.0712$ & $0.0373$ & $\mathbf{0.0411}$ \\
  Weighted & $0.0414$ & $0.0179$ & $0.0295$ & $0.0724$ & $0.0360$ & $\mathbf{0.0395} $ \\ 
  Min-Norm & $0.3429$ & $0.3460$ & $0.2475$ & $0.2608$ & $0.2012$ & $\mathbf{0.2797}$ \\
  `MM' & $0.1090$ & $0.1081$ & $0.0715$ & $0.1077$ & $0.0586$ & $\mathbf{0.0910}$ \\
  `LL' & $0.0952$ & $0.1271$ & $0.1324$ & $0.0969$ & $0.1344$ & $\mathbf{0.1172}$ \\
  Poisson & $0.0922$ & $0.0847$ & $0.0615$ & $0.0989$ & $0.0538$ & $\mathbf{0.0782}$ \\
  AWGN & $0.1571$ & $0.1806$ & $0.1723$ & $0.1369$ & $0.1425$ & $\mathbf{0.1579}$ \\
  
  \hline 
  KL & $ 0.3687$ & $0.3841$ & $0.3012$ & $0.2498$ & $0.2542$ & $\mathbf{0.3116}$ \\
  IS$^{(0)}$ & $0.3422$ & $0.3141$ & $ 0.2562$ & $ 0.1989$ & $0.1970$ & $\mathbf{0.2617}$ \\
  IS$^{opt}$ & $0.5184$ & $0.5454$ & $0.4044$ & $0.3547$ & $0.3191$ & $\mathbf{0.4284}$ \\
  \hline \hline
 
  \end{tabular}}

  \label{table:corr_coeffs_d5}

\end{table}

\begin{table}[!htbp]
\caption{Table of Correlations, $40 \times 40$ Images with $\alpha = 10$} % title of Table
\centering % used for centering table
{
\footnotesize
  \begin{tabular}{c | c | c | c | c | c | c }
 \hline \hline
  Method & Image $1$ & Image $2$ & Image $3$ & Image $4$ & Image $5$ & Average \\
  \hline % inserts single horizontal line

  Classical & $0.0392$ & $0.0164$ & $0.0284$ & $0.0715$ & $0.0349$ & $\mathbf{0.0381}$	\\
  Truncated & $0.0383$ & $0.0148$ & $0.0277$ & $0.0710$ & $0.0342$ & $\mathbf{0.0372}$ \\
  $\ell_2$-Loss & $0.0393$ & $0.0165$ & $0.0284$ & $0.0715$ & $0.0349$ & $\mathbf{0.0381}$ \\
  Orthogonal & $0.0438$ & $0.0207$ & $0.0312$ & $0.0737$ & $0.0372$ & $\mathbf{0.0413}$ \\
  Weighted & $0.0417$ & $0.0178$ & $0.0300$ & $0.0727$ & $0.0360$ & $\mathbf{0.0396} $ \\ 
  Min-Norm & $0.4439$ & $0.4584$ & $0.3367$ & $0.3389$ & $0.2802$ & $\mathbf{0.3716}$ \\
  `MM' & $0.0852$ & $0.0755$ & $0.0587$ & $0.0969$ & $0.0520$ & $\mathbf{0.0737}$ \\
  `LL' & $0.3104$ & $0.3685$ & $0.2855$ & $0.2222$ & $0.2179$ & $\mathbf{0.2809}$ \\
  Poisson & $0.0953$ & $0.0882$ & $0.0639$ & $0.1011$ & $0.0546$ & $\mathbf{0.0806}$ \\
  AWGN & $0.3159$ & $0.3637$ & $0.2703$ & $0.2220$ & $0.1977$ & $\mathbf{0.2739}$ \\
  
  \hline 
  KL & $ 0.5296$ & $0.5511$ & $0.4246$ & $0.3592$ & $0.3476$ & $\mathbf{0.4424}$ \\
  IS$^{(0)}$ & $0.4948$ & $0.5021$ & $ 0.3807$ & $ 0.3095$ & $0.2612$ & $\mathbf{0.3896}$ \\
  IS$^{opt}$ & $0.6349$ & $0.6658$ & $0.5078$ & $0.4456$ & $0.4172$ & $\mathbf{0.5342}$ \\
  \hline \hline
 
  \end{tabular}}

  \label{table:corr_coeffs_d10}

\end{table}

%\begin{table*}[!htbp]
%\caption{Table of correlations, average values over $40 \times 40$ images with $\alpha = 5$ (top), and $\alpha = 10$ (bottom).} % title of Table
%\centering % used for centering table
%{
%\footnotesize
%  \begin{tabular}{c | c | c | c | c | c | c | c | c | c | c | c | c}%p{0.1in}p{0.7in}p{1.4in}}
% \hline \hline
% Classical & Truncated & $\ell_2$-Loss & Orthogonal & Weighted & Min-Norm & `MM' & `LL' & Poisson & AWGN & KL & IS$^{(0)}$ & IS$^{opt}$ \\
%  \hline % inserts single horizontal line
%
%  ${0.0377}$ & ${0.0370}$ & ${0.0378}$ & ${0.0411}$ & ${0.0395} $ & ${0.2797}$ & ${0.0910}$ & ${0.1172}$ & ${0.0782}$ & ${0.1579}$  & $\mathbf{0.3116}$  & $\mathbf{0.2617}$ & $\mathbf{0.4284}$ \\
%  \hline
%   ${0.0381}$ & ${0.0372}$ & ${0.0381}$ & ${0.0413}$ & ${0.0396} $ & ${0.3716}$ & ${0.0737}$ & ${0.2809}$ & ${0.0806}$ & ${0.2739}$  & $\mathbf{0.4424}$  & $\mathbf{0.3896}$ & $\mathbf{0.5342}$ \\
%\hline \hline
%  \end{tabular}}
%
%  \label{table:corr_coeffs_d5}
%
%\end{table*}

As noted, the results on the $40 \times 40$ sized problem serve as a good benchmark to assess the robustness of the evaluated methods with respect to model errors and noise. 
This is due to the larger residual in estimating the sensing matrix of the problem, which is reflected by the ML-parameters estimated for the Poisson and Gaussian noise models. 
Regarding our formulations, we do not invoke any noise dependent processing, and essentially assess the robustness of our framework with the chosen information theoretic metrics. 
The correlation coefficients averaged over our the provided $5$ test images are summarized in Table \ref{table:corr_coeffs_d5}, at oversampling factors of $\alpha = 5$ and $\alpha = 10$ in the Monte Carlo simulation, in the first and second rows, respectively. 
Empirical average correlations clearly indicate that the model mismatch drastically impacts the classical method and its closely related counterparts, as well as orthogonality promoting methods. 
With the increased model error, Gaussian optimal methods also demonstrate significant decline in performance, where $MM$, $LL$, and its noise model optimal counterparts are all discouragingly low compared to the $16 \times 16$ case, with the Gaussian noise model providing the best average correlation with the benefit of accurately modeling the residuals. 
Overall, the experiments highlight $3$ of the methods as highly functional for a model that is non-Gaussian even in the presence of modeling errors: our optimal $KL$ and $IS$ divergence minimizing methods, and the minimum norm solution. 
Ultimately, the $LL$ optimal method shows significant improvement as well as its Gaussian noise model optimal method, and our $IS^{(0)}$ formulation as the number of measurements increase to $10$, which can be interpreted within their fundamental relation to the Itakura-Saito divergence metric. 
In the end, the key advantage of our approach compared to the minimum norm solution is its superior performance at lower oversampling factors, and overall computational complexity, as the latter involves the inversion of an $M \times M$ system of equations which dominates the procedure as $\alpha = M/N$ is increased.

\section{Conclusion}\label{sec:4_sec6}

In this paper, we undertake the spectral initialization as an estimation task, and formulate a framework for optimally designing spectral methods by leveraging information theoretic distortion metrics over phaseless measurements. 
We begin by deconstructing fundamental principles of spectral methods in the literature, and rigorously assess the statistical foundations, and successful heuristics within a geometric interpretation. 
Our approach leads to the identification of inherent mechanisms that are expected to promote good initial estimates in prominent methods from phase retrieval literature, which we formalize under an approximate loss minimization perspective where the sample processing function has the duty of warping the traditional $\ell_2$ alignment metric to more favorable similarity measures. 
Thereon, we build a general framework by utilizing Bregman divergences, which through universal properties facilitates a systematic approach to the design of spectral methods by the principle of Bregman representation and minimal distortion. 
Towards designing novel spectral methods, we consider well established information theoretic metrics from statistical signal processing and speech recognition literature, which naturally measure distortion between strictly positive information, namely, KL divergence with probability distributions, and Itakura-Saito divergence with power spectral densities. 
%Given the well known optimality of these information theoretic metric, we consider phaseless measurements as analogues 
We then consider estimation from phaseless measurements as analogues of these settings under proper normalization via strict positivity, and develop KL divergence and IS divergence minimizing spectral methods under our minimal Bregman representation principle. 
Furthermore, we establish that the recent work on optimal sample processing functions in the Gaussian model is an instance of our formulation, where global optimality in the Gaussian setting is implied due to asymptotic optimality of the Itakura Saito distortion in discriminating zero mean Gaussian processes. 
Finally, we present numerical simulations which demonstrate that our formalism consistently outperforms the methods in the literature, and support our theoretical on Gaussian sampling and real optical imaging settings.

\bibliographystyle{IEEEtran}
{
\bibliography{refs,ref1_alt,ref1_v2}}

\end{document}